\documentclass[a4paper, USenglish, cleveref, autoref, thm-restate,numberwithinsect]{lipics-v2021}
 
\usepackage{todonotes}
\usepackage{booktabs}
\usepackage{multirow}
\usepackage{array}
\usepackage{xcolor}

\newif\ifshort
\shortfalse 
\newcommand{\defproblem}[3]{
\vspace{3mm}
\noindent\fbox{
  \begin{minipage}{0.96\textwidth}
  \begin{tabular*}{\textwidth}{@{\extracolsep{\fill}}lr} \textsc{#1} \\ \end{tabular*}
  {\bf{Input:}} #2  \\
  {\bf{Question:}} #3
  \end{minipage}
  }
  \vspace{2mm}
}

\newcommand{\defparproblem}[4]{
\vspace{3mm}
\noindent\fbox{
  \begin{minipage}{0.96\textwidth}
  \begin{tabular*}{\textwidth}{@{\extracolsep{\fill}}lr} \textsc{#1} \\ \end{tabular*}
  {\bf{Input:}} #2  \\
  {\bf{Parameter:}} #3\\
  {\bf{Question:}} #4
  \end{minipage}
  }
  \vspace{2mm}
}

\newcommand{\CFCut}{{\sc Conflict-Free Edge Cut}}
\newcommand{\mincut}{{\sc Min Conflict-Free Edge Cut}}

\newcommand{\cfcut}{{\sc CF-Cut}}
\newcommand{\mincfcut}{{\sc Min-CF-Cut}}
\newcommand{\ThreeColor}{{\sc $3$-Colorability}}
\newcommand{\nphard}{\textsf{NP-hard}}
\newcommand{\yesinst}{{\sf yes}-instance}
\newcommand{\noinst}{{\sf no}-instance}

\newcommand{\intra}{intra-partite}
\newcommand{\inter}{inter-partite}
\newcommand{\MBC}{{\sc Multicolored Biclique}}
\newcommand{\UnitMIS}{{\sc Unit 2-Track Multicolored Independent Set}}
\newcommand{\gen}{{\textsf{general}}}
\newcommand{\nka}{{{\rm NP} $\subseteq$ {\rm coNP/poly}}}
\newcommand{\cOO}{{\mathcal O}}

\title{On the Parameterized Complexity of Conflict-Free Edge Cut in Undirected Graphs} 

\author{Sourav Das}{Indraprastha Institute of Information Technology Delhi, New Delhi, India}{souravd@iiitd.ac.in}{}{}

\author{Ashwin Jacob}{National Institute of Technology Calicut, Kozhikode, India \and \url{https://ashwinjacob.github.io}}{ashwinjacob@nitc.ac.in}{https://orcid.org/0000-0003-4864-043X}{}

\author{Arpit Kumar}{Indraprastha Institute of Information Technology Delhi, New Delhi, India \and \url{https://sites.google.com/iiitd.ac.in/arpit-gang}}{arpitk@iiitd.ac.in}{https://orcid.org/0009-0002-0741-5682}{}

\author{Diptapriyo Majumdar}{Indraprastha Institute of Information Technology Delhi, New Delhi, India \and \url{https://diptapriyomajumdar.wixsite.com/toto}}{diptapriyo@iiitd.ac.in}{https://orcid.org/0000-0003-2677-4648}{}

\authorrunning{S. Das et al.} 

\Copyright{S. Das, A. Jacob, A. Kumar, D. Majumdar} 



\category{} 

\relatedversion{} 

\nolinenumbers 

\EventEditors{John Q. Open and Joan R. Access}
\EventNoEds{2}
\EventLongTitle{42nd Conference on Very Important Topics (CVIT 2016)}
\EventShortTitle{CVIT 2016}
\EventAcronym{CVIT}
\EventYear{2016}
\EventDate{December 24--27, 2016}
\EventLocation{Little Whinging, United Kingdom}
\EventLogo{}
\SeriesVolume{42}
\ArticleNo{23}

\begin{document}

\maketitle

\begin{abstract}
In this paper, we study {\CFCut} which is a recently introduced conflict-free version of the {\sc Min-Cut} problem that asks to find the minimum number of edges to disconnect a connected graph.
The {\CFCut} ({\cfcut} in short) takes as input a connected undirected graph $G = (V, E)$, a conflict graph $\widehat{G}$ such that $E(G) = V(\widehat{G})$, and the objective is to decide whether there exists $F \subseteq E(G)$ such that $G - F$ is disconnected and $F$ is an independent set in $\widehat{G}$.
Rauch, Rautenbach, and Souza [IPL-2025] proved that {\CFCut} is NP-Complete and also provided some results on the parameterized complexity of {\CFCut}.
A related variant {\mincut} ({\mincfcut} in short) takes a connected graph $G$, a conflict graph $\widehat{G}$ such that $V(\widehat{G}) = E(G)$, and an integer $k$ as input and asks if there is a set $F$ of at most $k$ edges such that $G - F$ is disconnected and $F$ is an independent set in $\widehat{G}$.
In this paper, we extend the work by Rauch et al. [IPL-2025], and provide a systematic study on the {\CFCut} and {\mincfcut} from the perspective of parameterized complexity and polynomial kernelization.
We prove that {\CFCut} is NP-hard when the dissociation number (${\sf diss}$) of the input graph is at most two.
We also complement it by proving that {\CFCut} is polynomial-time solvable when the dissociation number of an input graph is at most one.
Additionally, for {\mincfcut}, we consider both solution size ($k$), as well as various structural parmaeters of the input graph as parameters, and provide fixed-parameter tractability and W[1]-hardness results when the conflict graph $\widehat{G}$ is restricted to various graph classes.
We also prove that unless {\nka}, {\mincfcut} admits no polynomial kernel when parameterized by the vertex cover number (${\sf vc}$) of the input graph; and also when parameterized by the sum of the solution size ($k$) and the vertex integrity (${\sf vi}$) of the input graph.
\end{abstract}



\section{Introduction}
\label{sec:intro}


Graph modification, graph protection, and graph separation problems are among the fundamental problems and have been widely explored from the perspective of classical and parameterized complexity.
In the recent years, several graph theoretic and combinatorial optimization problems have been explored by introducing {\em conflict constraints} --
restrictions ensuring that certain pairs of objects, vertices, or edges cannot be chosen simultaneously, typically to model real-world incompatibilities such as resource clashes or scheduling interference.
Introduction of conflict constraints can add an additional layer in the ``computational complexity'' of a combinatorial optimization problem.
This framework has been applied to various classical problems e.g. {\sc Knapsack} \cite{BettinelliCM17}, {\sc Shortest Path, Maximum Matching} \cite{DarmannPSW11}, {\sc Spanning Tree} \cite{DarmannPSW11,ZhangKP2011}.
Therefore, it is natural to ask -- ``{\em how does the complexity of an optimization problem change when conflict constraints are introduced?}''
From the perspective of parameterized complexity, the question is relevant for various reasons.
If a parameterized problem $L \subseteq \Sigma^* \times \mathbb{N}$ is FPT with respect to some parameter, then does it remain so once conflict constraints are added, or does it become W-hard?
In the latter case, then under what structural restrictions on the conflict graph can tractability be restored?
Studying conflict-free variants of parameterized problems in this way yields a precise map of which conflict-graph classes preserve tractability.

Despite this general interest, the conflict-free variant of graph separation problems has received comparatively little attention.
In this paper, we study {\CFCut} which asks to find the existence of an edge cut with conflict constraints on the edge set.
Rauch et al. \cite{RauchRS25} studied {\CFCut} from the classical and parameterized complexity perspective by considering the vertex cover number, clique cover number, and feedback vertex number of the primary graph and clique cover number of the conflict graph as parameters.
However, there are several structural parameters of the primary graph other than the vertex cover number and feedback vertex number, as well as various structural restrictions on the conflict graphs are unexplored.
In this paper, we study the {\CFCut} problem with respect to various natural and structural parameters, and graph classes, extending the parameterized complexity landscape of this problem.

A {\em cut} in a connected undirected graph $G = (V, E)$ is a partition $(A, B)$ of $V$ into nonempty parts. 
The {\em edge cut} of $(A, B)$, denoted by $E_G(A, B)$ is the set of edges $uv$ such that $u \in A$ and $v \in B$.
A cut $(A, B)$ of $G$ is a {\em matching cut} if every $u \in A$ has at most one neighbor in $B$ and every $u \in B$ has at most one neighbor in $A$.
The decision problem {\sc Matching Cut} asks if a connected undirected graph has a matching cut.
{\sc Matching Cut} is NP-Complete in various special classes of graphs, and it is well-studied both in classical  \cite{FeghaliLPR25,LuckePR23,LuckePR24} and parameterized complexity \cite{AravindS21,KomusiewiczKL20}.
Observe that in {\sc Matching Cut}, the edge cut of the corresponding cut is required to form a matching - i.e., a collection of edges no two of which share an endpoint.
Recently, Rauch et al. \cite{RauchRS25} generalized this notion and defined {\CFCut} as follows.

\defproblem{{\CFCut} ({\cfcut})}{A connected undirected graph $G = (V, E)$ and a conflict graph $\widehat{G}$ such that $V(\widehat{G}) = E(G)$.}{Does there exist a subset $F \subseteq E(G)$ such that $F$ is an edge cut of $G$ and is independent in $\widehat{G}$?}

In the above definition, we say that $G$ is the {\em primary graph} and $\widehat{G}$ is the {\em conflict graph}.
Observe that the {\sc Matching Cut} problem is the same as {\CFCut} problem while the conflict graph is the line graph of the primary graph.
Additionally, note that the definition of {\CFCut} asks whether there exists an edge cut of the primary graph that is independent in the conflict graph, but does not require the size of the edge cut to be minimized.
To distinguish between these two objectives, we define the following problem variant, in which we additionally wish to minimize the number of edges crossing the conflict-free cut.

\defproblem{{\mincut} ({\mincfcut})}{A connected undirected graph $G = (V, E)$, a conflict graph $\widehat{G}$ such that $V(\widehat{G}) = E(G)$, and an integer $k$.}{Is there $F \subseteq E(G)$ such that $F$ is an edge cut of $G$, $|F| \leq k$ and $F$ is independent in $\widehat{G}$?}

It is important to observe that {\mincfcut} is as hard as {\cfcut}.
To see this, observe that if {\mincfcut} is polynomial-time solvable by an algorithm $\mathcal{A}$, then we run $\mathcal{A}$ for $k =|E(G)|$.
Subsequently, if the algorithm $\mathcal{A}$ answers yes, then we output that $(G, \widehat{G})$ has a conflict-free edge cut, and otherwise $(G, \widehat{G})$ does not have a conflict-free edge cut.
If $\widehat{G}$ is an edgeless graph, then {\mincfcut} is the same as the {\sc Min-Cut} problem which is polynomial-time solvable \cite{diestel2025graph}.
Motivated by the works on the classical and parameterized complexity of {\CFCut} reported in \cite{KalyanasundaramK26,RauchRS25}, together with the broader line of work on conflict-free variants of combinatorial \cite{gupta2026computational} and graph theoretic problems \cite{AgrawalJKL018,AgrawalJKS20,JacobMR21,BhyravarapuHHKR24,BodlaenderKP2021,JainKM20}, we initiate a systematic study of {\CFCut} and {\mincut} from the perspective of parameterized complexity and kernelization.
As {\cfcut} asks to check the existence of a conflict-free edge cut in a graph, the natural choices of parameters for {\cfcut} 
are various {\em structural parameters of the primary graph}; e.g. vertex cove number (${\sf vc}$), dissociation number (${\sf diss}$), vertex integrity (${\sf vi}$) etc.
However {\mincfcut}, we ask to check if there exists a conflict-free edge cut of size at most $k$.
Hence, the choice of parameters are of two types: (i) one natural choice of parameter is $k$ (the solution size), and (ii) the alternate possible choices are various structural parameters of the input graph.
In the case of {\CFCut}, structural parameters are natural parameter choices.

\subparagraph*{Our Contribution.} Rauch et al. \cite{RauchRS25} have initiated a systematic study on the {\CFCut} when parameterized by solution size and also with respect to structural parameters.
In particular, one of their results impled that {\mincfcut} is W[1]-hard when the size of the cut is considered as the parameter.
Additionally, they proved that {\CFCut} is FPT when parameterized by the vertex cover number (${\sf vc}$) of the primary graph and para-NP-hard when parameterized by the  deletion distance to a graph of degree at most two.
Our first result closes this gap by proving that {\cfcut} is para-NP-hard when the dissociation number (${\sf diss}$) of the primary graph is considered as the parameter.
More specifically, we prove the following result.

\begin{restatable}{theorem}{DissNPhard}
\label{thm:diss-hardness}
{\CFCut} is NP-hard when the primary graph has dissociation number (${\sf diss}$) at most two.
Additionally, {\CFCut} is polynomial-time solvable when the primary graph has dissociation number at most one.
\end{restatable}


Our next result first proves that {\mincfcut} is fixed-parameter tractable but admits no polynomial kernel unless {\nka} when the solution size ($k$) is considered as the parameter and the conflict graph is $\alpha$-degenerate for some fixed constant $\alpha$.

\begin{restatable}{theorem}{SolSizeFPTMinCFC}
\label{thm:degen-FPT-sol-size}
{\mincfcut} when parameterized by $k$  is fixed-parameter tractable when the conflict graph is $\alpha$-degenerate for some fixed constant $\alpha$.
\end{restatable}

\begin{restatable}{theorem}{SolSizeNoKernelCFC}
\label{thm:degen-no-kernel-sol-size}
Unless \nka, {\mincut} admits no polynomial compression when parameterized by the sum of the solution size $k$ and the vertex integrity of the primary graph even when the primary graph is 3-degenerate and the conflict graph is 1-regular.
\end{restatable}


Finally, we also consider the cases for {\mincfcut} when the solution size ($k$) is the parameter and the conflict graph is bipartite and also the case when the conflict graph is a unit interval graph.
In particular, we prove the following two results for {\mincfcut}.

\begin{restatable}{theorem}{UnitIntervalHard}
\label{thm:hardness-unit-interval}
    {\mincfcut} parameterized by the solution size $k$ is $W[1]$-hard when conflict graph is unit interval graph.
\end{restatable}

\begin{restatable}{theorem}{BipartiteHardness}
\label{thm:bipartite-Whardness}
{\mincfcut} parameterized by the solution size $k$ is W[1]-hard when the conflict graph is bipartite.
\end{restatable}

Finally, we also prove the following kernelization lower bound {\mincut}.

\begin{restatable}{theorem}
{MinCutVC}
\label{thm:min-cut-vc}
Unless {\nka}, {\mincfcut} parameterized by the vertex cover number admits no polynomial compression. 	
\end{restatable}

\begin{table}[t]
  \centering
  \begin{tabular}{l|cccc}
    \toprule
    \multirow{2}{*}{\textbf{(primary, conflict)}} & \multicolumn{2}{c}{\textbf{\cfcut}} & \multicolumn{2}{c}{\textbf{\mincfcut}} \\
    \cmidrule(lr){2-3} \cmidrule(lr){4-5}
      & Parameter & Result & Parameter & Result \\
    \midrule
    ({\gen}, {\gen}) & ${\sf diss}$ & para-NP-hard  & ${\sf diss}$ & para-NP-hard \\
    ({\gen}, {\gen}) & ${\sf vc}$ & FPT \cite{RauchRS25} & ${\sf vc}$ & No poly kernel \\
    ({\gen}, {\sf bipartite}) & -- & -- & ${k}$ & W[1]-hard \\
    ({\gen}, {\sf $\alpha$-degen}) & -- & -- & $k$ & FPT for fixed $\alpha$ \\
    ({\sf 3-degen}, {\sf 1-reg}) & -- & -- & $k + {\sf vi}$ & No poly kernel \\
    ({\gen}, {\sf prop-int}) & -- & -- & $k$ & W[1]-hard\\
    \bottomrule
  \end{tabular}
  \caption{Summary of Our Results.}
  \label{table:results-summary}
\end{table}


\subparagraph*{Related Work.} 
Recently, Kalyanasundaram and Kumar \cite{KalyanasundaramK26} have provided some structural and classical complexity results on {\CFCut}, subsequent to the first set of algorithmic results on the same problem by Rauch et al. \cite{RauchRS25}.
Banik et al. \cite{BanikFG0J024} have considered different variants of cut problems while being restricted to matroid constraints only.
Some of the other combinatorial problems have been extensively studied with respect to conflicts, e.g. {\sc Maximum Matching, Shortest Path} \cite{AgrawalJKS20,DarmannPSW11},  {\sc Minimum Spanning Tree} \cite{DarmannPSW11,ZhangKP2011}, {\sc Feedback Vertex Set} \cite{AgrawalJKL018}, {\sc Vertex Cover} \cite{JainKM20}, {\sc Set Cover} \cite{JacobMR21}, {\sc Knapsack} \cite{BettinelliCM17}, {\sc Bin Packing} \cite{CapuaFLV18,GendreauLS2004}, {\sc Scheduling} and its related variants \cite{BodlaenderHvL2025,BodlaenderJW94,HongBertrand2018, MallekBoudhar2024,TellacheB2018,ZhangZCCW2023}.
Some of previous papers have explored other variants of separation problems, e.g. {\sc Matching Cut} \cite{FeghaliLPR25,KomusiewiczKL20,LuckePR23,LuckePR24}, {\sc $d$-Cut} \cite{GomesS21}, {\sc Max-Cut} \cite{FariaKSSS20}, {\sc Multiway Cut} \cite{ChenLL2009,ChitnisCHPP16,CyganPW2013}, {\sc Multicut} \cite{MarxRazgon2014} but their conflict-free variant has remained unexplored.

\section{Preliminaries and Notations}
\label{sec:prelims}

\subparagraph*{Sets, Numbers, and Graph Theory.} 

Throughout the paper, we use undirected graphs.
Given a graph $G = (V, E)$, a set $S \subseteq V$ is {\em dissociation set} if maximum degree of $G - S$ is at most 1.
The {\em dissociation number} of a graph is the number of vertices in a dissociation set with the smallest cardinality.
Let $(A,B)$ denote a cut of $G$.
Two vertices $u$ and $v$ are \emph{\intra} vertices with respect to $(A, B)$, if either $\{u,v\} \subseteq A$ or $\{u,v\} \subseteq B$.
Conversely, $u$ and $v$ are called \emph{\inter} vertices with respect to $(A, B)$, if $u$ and $v$ are separated by $(A, B)$; equivalently $\{u, v\} \cap A, \{u, v\} \cap B \neq \emptyset$.
The same nomenclature extends naturally to subsets of vertices.
Given a graph $G=(V, E)$ and a set $X\subseteq V(G)$, we define the {\em identification of the vertex subset $X$ into a new vertex $u_{X}$} by constructing $G'$ as follows: first, delete the vertices in $X$ from $G$ and add a new vertex $u_{X}$.
Then, for every $v\in N_{G}(X)$, make $vu_{X}$ an edge of $G'$.
Similar notion of graph operation for $i=1$ has been defined in the literature \cite{comas_vertex_2009,diestel2025graph,majumdar_approximate_2020,morelle_vertex_2026,golovach2026identification}.
Additionally, given an integer $p$, a collection of edge-disjoint paths between two vertices $x$ and $y$ is called a {\em $p$-uncuttable gadget} between $x$ and $y$, if there are exactly $p+1$ edge-disjoint paths between $x$ and $y$ each having exactly 2 edges. 

\begin{lemma}
\label{lemma:triangle-double-edge-monochromatic}
Let $G$ be a (multi)graph and $\widehat{G}$ be a conflict graph such that $V(\widehat{G}) = E(G)$.
If there are two edges between $x, y \in V(G)$ and those two edges are adjacent in $\widehat{G}$, then $x$ and $y$ are \intra\ with respect to every conflict free cut of $G$.
Similarly, if there is a triangle $T = \{x, y, z\}$ in $G$ such that $xy, yz, xz$ are pairwise adjacent in $\widehat{G}$, then $T$ is \intra\ with respect to every conflict-free cut of $G$.
\end{lemma}

\begin{proof}
    First, let $F$ be any conflict free edge cut and $(A,B)$ be the partition corresponding to $F$.
    We first give the proof of the first item. Let $e_1, e_2$ be the two edges between $x$ and $y$ such that $\{e_1,e_2\} \in E(\widehat{G})$.  
    Suppose for the sake of contradiction that $x$ and $y$ are inter-partite. Then, without loss of generality, we can assume that $x \in A$ and $y \in B$. Consequently, $e_1, e_2 \in F$ which contradicts the face that $F$ is independent in $\widehat{G}$. Hence, $x$ and $y$ are \intra\ with respect to every conflict free cut of $G$.

    Now for the second item, let $T = \{x, y, z\}$ induces a triangle in $G$ such that $xy, yz, xz$ are pairwise adjacent in $\widehat{G}$. Suppose for the sake of contradiction that $T$ is \inter. Then, without loss of generality, we can assume that $x, y \in A$ and $z \in B$. Consequently, $xz, yz \in F$ which contradicts the fact that $F$ is independent in $\widehat{G}$ as $\{xz, yz \} \in E(\widehat{G})$. Hence, $T$ is \intra\ with respect to every conflict-free cut of $G$.
\end{proof}

\begin{definition}[$k$-Independence Covering Family, \cite{LokshtanovPSSZ20}]
\label{defn:k-ind-family}
    For a graph $G$ and $k \in \mathbb{N}$, a family of independent sets of $G$ is called a {\em $k$-independence covering family} for $(G, k)$, denoted by $\mathcal{F}(G, k)$, if for any independent set $X$ in $G$ of size at most $k$, there exists $Y \in \mathcal{F}(G, k)$ such that $X \subseteq Y$.
\end{definition}

\begin{proposition}[Deterministic Independence Covering Lemma, \cite{LokshtanovPSSZ20}]
\label{lemma:k-ind-covering}
    There is an algorithm that given an $\alpha$-degenerate graph $G$ and $k\in \mathbb{N}$, runs in time $\cOO(\binom{k^2 {(\alpha+1)}^2}{k} \cdot {(k(\alpha+1))}^{\cOO(1)} \cdot (n+m) \log n)$ and outputs a $k$-independence covering family for $(G, k)$ of size at most $\binom{k^2{(\alpha+1)}^2}{k} \cdot {(k(\alpha+1))}^{\mathcal{O}(1)} \cdot \log n$.
\end{proposition}

\subparagraph*{Parameterized Complexity and Kernelization.} A parameterized problem $L$ is a set of instances
$(x, k) \in \Sigma^{*}\times \mathbb{N}$ where $\Sigma$ is a finite alphabet and $k \in \mathbb{N}$ is a parameter.
A parameterized problem $L \subseteq \Sigma^* \times \mathbb{N}$ is fixed-parameter tractable (or FPT) if given $(x, k) \in \Sigma^{*}\times \mathbb{N}$, there is an algorithm $\mathcal{A}$ that correctly decides if $(x, k) \in L$ in $f(k)|x|^{\mathcal{O}(1)}$-time for some computable function $f : \mathbb{N}\rightarrow \mathbb{N}$. This algorithm $\mathcal{A}$ is called fixed-parameter algorithm (or FPT algorithm) for the problem $L$.
Observe that we allow combinatorial explosion in the parameter $k$ while the algorithm runs in polynomial-time with respect to $|x|$.
We say that two instances $(x, k)$ of $L$ and $(x', k')$ of $L$ are equivalent when $(x, k) \in L$ if and only if $(x', k') \in L$.
The notion of kernelization (preprocessing) for in the area of parameterized complexity is defined as follows.
A parameterized problem $L \subseteq \Sigma^{*} \times \mathbb{N}$ is said to admit a kernelization if there exists an algorithm that, given an instance $(x, k)$ of $L$, outputs in time polynomial in $|x| + k$ an equivalent instance $(x', k')$ of $L$ such that $|x'| + k' \leq g(k)$ for some computable function $g : \mathbb{N} \rightarrow \mathbb{N}$. The instance $(x', k')$ is called a kernel for the instance $(x, k)$. If $g(k)$ is $k^{\mathcal{O}(1)}$, then we say that $L$ admits a polynomial kernel.
%
A kernelization algorithm usually consists of a collection of reduction rules that have to be applied exhaustively. A reduction rule is called safe if given an instance $(x, k)$ of $L$, one application of the reduction rule outputs an equivalent instance $(x', k')$ of $L$. 
It is well-known due to \cite{downey2013fundamentals,cygan2015parameterized,fomin2019subquadratic} that a decidable parameterized problem is FPT if and only if it admits a kernelization.

\begin{definition}[Parameterized Reduction, \cite{cygan2015parameterized}]
\label{defn:parameterized-reduction}
Let $A, B \subseteq \Sigma^{*}\times \mathbb{N}$ be two parameterized problems. A parameterized reduction from $A$ to $B$ is an algorithm that, given an instance $(x, k)$ of $A$, outputs an instance $(x', k')$ of $B$ such that 
\begin{itemize}
    \item $(x, k)$ is a {\yesinst} of $A$ if and only if $(x', k')$ is a {\yesinst} of $B$,
    \item $k' \leq g(k)$ for some computable function $g : \mathbb{N} \rightarrow \mathbb{N}$, and
    \item the running time is $f(k) \cdot |x|^{\cOO(1)}$ for some computable function $f : \mathbb{N} \rightarrow \mathbb{N}$.
\end{itemize}
\end{definition}

\subparagraph*{W-hierarchy.} Within parameterized problems, the W-hierarchy may be seen as the parameterized equivalent to the class NP of classical decision problems. Without entering into details (see \cite{cygan2015parameterized,downey2013fundamentals} for the formal definitions), a parameterized problem being W-hard can be seen as a strong evidence that this problem is not FPT.  Parameterized complexity also has a completeness theory: many parameterized problems that appear not
to be fixed-parameter tractable have been classified as being complete under
parameterized reductions for the complexity classes in the so-called
weft hierarchy $\text{FPT} \subseteq \text{W[1]} \subseteq \text{W[2]} \subseteq \ldots \subseteq \text{XP}$,
where XP is the class of problems solvable in $\cOO(|x|^{f(k)})$ time for some computable function $f$.


\begin{definition}[Polynomial Equivalence Relation, \cite{bodlaender_cross-composition_2010}]
\label{defn:poly-eqv-relation}
    An equivalence relation $\mathcal{R}$ on $\Sigma^{*}$ is called a {\em polynomial equivalence relation} if the following two conditions hold:  
    \begin{itemize}
        \item There is an algorithm that given two strings $x, y \in \Sigma^{*}$ decides whether $x$ and $y$ belong to the same equivalence class in $(\lvert x\rvert + \lvert y\rvert)^{\cOO(1)}$-time. 
        \item For any finite set $S\subseteq \Sigma^{*}$ the equivalence relation $\mathcal{R}$ partitions the elements of $S$ into at  most $(\max\limits_{x\in S} \lvert x\rvert)^{\cOO(1)}$ equivalence classes.
    \end{itemize}
\end{definition}

\begin{definition}[OR-Cross-Composition, \cite{bodlaender_cross-composition_2010}]
\label{defn:OR-cross-composition}
    Let $L\subseteq \Sigma^{*}$ be a language and let $Q\subseteq \Sigma^{*}\times \mathbb{N}$ be a parameterized problem.
    We say that $L$ {\em cross-composes} into $Q$ if there is a polynomial equivalence relation $\mathcal{R}$ and an algorithm $\mathcal{A}$ which, given $2^t$ strings $x_1, x_2, \dots, x_{2^t} \in \Sigma^*$ that are equivalent with respect to $\mathcal{R}$, computes an instance $(x^*, k^*)\in \Sigma^{*}\times \mathbb{N}$ in time polynomial in $\sum\limits_{i=1}^{2^t} \lvert x_i\rvert$  such that:  
    \begin{itemize}
        \item $(x^*, k^*) \in Q$ if and only if $x_i \in L$ for some $1 \le i \le 2^t$, 
        \item $k^*$ is bounded by a polynomial in $\max\limits_{i\in [2^t]}\lvert x_i\rvert + {t}$.
    \end{itemize}
\end{definition}

\begin{proposition}[Corollary 10, \cite{bodlaender_cross-composition_2010}]
\label{prop:cross-compo}
    If some set $L\subseteq \Sigma^*$ is \nphard\ under Karp reduction and $L$ cross-composes into the parameterized problem $Q$ then there is no polynomial compression for $Q$ unless $\mathsf{NP}\subseteq \mathsf{coNP/poly}$.
\end{proposition}

\begin{definition}[Polynomial Parameter Transformation, \cite{cygan2015parameterized}]
    \label{defn:poly-para-transformation}
    Let $P, Q\subseteq \Sigma^*\times \mathbb{N}$ be two parameterized problems.
    An algorithm $\mathcal{A}$ is called a {\em polynomial parameter transformation} (PPT) from $P$ to $Q$ if, given an instance $(x, k)$ of problem $P$, $\mathcal{A}$ works in polynomial time and outputs an equivalent instance $(\widehat{x}, \widehat{k})$ of problem $Q$, i.e. $(x, k)\in P$ if and only if $(\widehat{x}, \widehat{k})\in Q$ such that $\widehat{k}\le p(k)$ for some polynomial $p(\cdot)$. 
\end{definition}

\begin{proposition}[\cite{cygan2015parameterized}]
    \label{prop:no-poly-compression}
    Let $P,Q \subseteq \Sigma^*\times \mathbb{N}$ be two parameterized problems and assume there exists a polynomial parameter transformation from $P$ to $Q$.
    Then, if $P$ does not admit a polynomial compression, neither does $Q$.
\end{proposition}

\subparagraph*{Some Elementary Results.}
In this section, we want to explain some elementary result that states a precise relation between the fixed-parameter tractability of 
{\cfcut} and {\mincfcut}.

\ifshort
\begin{proposition}[$\star$]
\else
\begin{proposition}
\fi
\label{prop:FPT-CF-Cut-Min-CF-Cut}
Suppose that {\mincfcut} is FPT when parameterized by $\kappa$, then {\cfcut} is FPT when parameterized by $\kappa$.
\end{proposition}

\ifshort 
\else
\begin{proof}
Suppose that {\mincfcut} is FPT when parameterized by a parameter $\kappa$.
Let $(G, \widehat{G}, k)$ be an instance of {\mincfcut} parameterized by $\kappa$.
Since the problem is FPT, we can solve it in $f(\kappa) \cdot |G|^{\cOO(1)}$ time for some computable function $f:\mathbb{N} \rightarrow \mathbb{N}$. 
Hence, there exists an algorithm $\mathcal{A}$ that takes $(G, \widehat{G}, k)$, runs in  and outputs correctly whether $G$ has a conflict-free edge cut of at most $k$ edges.
We invoke this algorithm $\mathcal{A}$ for every $k = 1,2,\ldots,m$ in this order, where $m$ is the number of edges in $G$.
Note that $\mathcal{A}$ runs in $f(\kappa) \cdot |G|^{\cOO(1)}$ time for every $k = 1,\ldots,m$.
The smallest value $k$ for which we see that $\mathcal{A}$ outputs {\yesinst}, we return that $(G, \widehat{G})$ is a {\yesinst} for {\cfcut}.
If for each of the values $k$, we see that $\mathcal{A}$ outputs a {\noinst}, then we output that $(G, \widehat{G})$ is a {\noinst} for {\cfcut}.
Observe that our algorithm invokes $\mathcal{A}$ as a subroutine $m$ times.
Hence, the running time of our algorithm is $f(\kappa)\cdot m \cdot n^{\cOO(1)}$.
Therefore, {\cfcut} is FPT when parameterized by $\kappa$.
\end{proof}
\fi

\section{Fixed-Parameter Tractability Results}
\label{sec:solution-size-results}


In this section, we describe an FPT algorithm for {\mincfcut} when the conflict graph is $\alpha$-degenerate for a fixed constant $\alpha$.
Given a graph $G=(V,E)$, a vertex subset $S \subseteq E(G)$ and $k \in \mathbb{N}$, we define a weight function $w_S: E(G) \rightarrow [k]$ as follows.

$w_S(e) = \begin{cases} 1 & \text{if } e \in S \\[4pt] k+1 & \text{otherwise} \end{cases}$

\begin{lemma}
\label{lem:weight-cut}
Let $G=(V,E)$ be a graph, $S \subseteq E$ and $k \in \mathbb{N}$.
Then $G$ has a edge cut $F \subseteq S$ of size at most $k$ if and only if $(G,w_S)$ has an edge cut of total weight at most $k$.
\end{lemma}

\begin{proof}
For the forward direction ($\Rightarrow$), let $F$ be an edge cut in $G$ such that $F \subseteq S$ and $|F| \leq k$.
Since all the edges in $F$ have weight $1$, the total weight of edges in $F$ is at most $k$.
Thus, $(G,w_S)$ has an edge cut of total weight at most $k$.

In the backward direction ($\Leftarrow$), let $F'$ be a (minimum) edge cut of total weight at most $k$ for $(G,w_S)$.
Note that all the edges in $E \setminus S$ have weight $k+1$.
Hence none of these edges are in $F'$.
Thus, $F' \subseteq S$.
Since all the edges in $S$ have weight $1$ in $(G, w_S)$, the number of such edges in $F' \cap S$ is also at most $k$.
Thus $|F'| \leq k$, completing the proof.
\end{proof}

{\SolSizeFPTMinCFC*}


\begin{proof}
Let $(G, \widehat{G}, k)$ be an instance of {\mincfcut}, where $G=(V,E)$ is the connected primary graph, $\widehat{G}=(E,\widehat{E})$ is an $\alpha$-degenerate conflict graph, and $k$ is the maximum allowed size of the cut. 
Without loss of generality, we assume that $k \geq 0$.
The algorithm proceeds by first applying \Cref{lemma:k-ind-covering} to $\widehat{G}$ with parameter $k$. 
This generates a $k$-independence covering family $\mathcal{F}$ of vertex subsets of $\widehat{G}$ (which correspond to edge subsets of $G$).
We then iterate over each set $S \in \mathcal{F}$ one by one.
For a given set $S$, we define a weight function $w_S$ where we assign weight $1$ to every edge in $S$ and weight $k+1$ to every edge in $E \setminus S$.
We then compute a minimum weighted edge cut $F_S$ of $G$ in polynomial time using Stoer Wagner algorithm \cite{StoerWagner1997}.
If $w_S(F_S) \leq k$, the we output that $(G, \widehat{G}, k)$ is a {\yesinst} and stop. 
If no $S \subseteq \mathcal{F}$ has $w_S(F_S) \leq k$, we return that $(G, \widehat{G}, k)$ is a {\noinst}. 

Next, we prove the correctness of the algorithm by showing that $(G, \widehat{G}, k)$ is a {\yesinst} if and only if the algorithm outputs {\yesinst}.
For the forward direction ($\Rightarrow$), assume $(G, \widehat{G}, k)$ is a {\yesinst}.
Then there exists an edge cut $F \subseteq E$ such that $|F| \le k$ and $F$ is an independent set in $\widehat{G}$.
By the properties of the $k$-independence covering family, 
there exists $S \in \mathcal{F}$ such that $F \subseteq S$.
When the algorithm iterates over this specific $S$, \Cref{lem:weight-cut} guarantees that the edge cut of $(G,w_S)$ with minimum weight is at most $k$.
Hence the algorithm outputs {\yesinst}.

For the backward direction ($\Leftarrow$), assume the algorithm outputs {\yesinst}.
 This implies that for some $S \in \mathcal{F}$, the global minimum cut $F_S$ satisfies $w_S(F_S) \le k$.
 By \Cref{lem:weight-cut}, we have $F_S \subseteq S$ and $|F_S| \le k$.
 Since $S$ is an independent set in $\widehat{G}$, $F_S \subseteq S$ is also independent in $\widehat{G}$.
 Since $F_S$ is a cut in $G$, we can conclude that it is a conflict-free cut of size at most $k$.
 Thus, $(G, \widehat{G}, k)$ is a yes instance.

For the proof of the claimed running time, notice that applying the \Cref{lemma:k-ind-covering} to the $\alpha$-degenerate graph $\widehat{G}$ and parameter $k$ generates the family $\mathcal{F}$ of size at most $f(k, \alpha) \cdot |E|^{O(1)}$, and this family can be computed in time bounded by $f(k, \alpha) \cdot |E|^{O(1)}$.
Note that this $f(k, \alpha)$ is $\binom{k^2{(\alpha+1)}^2}{k} \cdot {(k(\alpha+1))}^{\mathcal{O}(1)}$.
For each set $S \in \mathcal{F}$, assigning the weights and finding a minimum weighted edge cut takes polynomial time.
Since the number of edges $|E|$ is bounded by $n^2$, the term $|E|^{O(1)}$ is bounded by $n^{O(1)}$. 
Therefore, the worst-case running time is strictly dominated by multiplying the size of the covering family by the polynomial time taken per branch.
Thus, the running time is $f(k, \alpha) \cdot n^{O(1)}$ and since $\alpha$ is a fixed constant, this provides an FPT algorithm when parameterized by $k$.
\end{proof}

In the next sentence, we prove that when the number of maximal independent sets of the conflict graph is polynomial in the input size, {\mincfcut} is polynomial-time solvable.

\begin{corollary}
    \label{cor:hereditary-conflict-graph}
Let $\mathcal{G}$ be a hereditary graph class such that for every graph $G \in \mathcal{G}$ with $n$ vertices, the collection of all maximal independent sets of $G$ is computable in $(|V(G)| + |E(G)|)^{O(1)}$ time.
Then the {\mincfcut} is polynomial-time solvable when the conflict graph is in $\mathcal{G}$.
\end{corollary}

\begin{proof}
    Let $(G, \widehat{G}, k)$ be an instance of {\mincfcut}, where $G=(V, E)$ is the connected primary graph, and $\widehat{G}=(E,\widehat{E})$ is the conflict graph in $\mathcal{G}$.
    The algorithm proceeds by first computing the collection $\mathcal{F}$ of all maximal independent sets of $\widehat{G}$ in polynomial time. 
    Note that every independent set is a subset of at least one maximal independent set.
    Hence, if $G$ has a conflict-free edge cut $F$, then there exists $S \in \mathcal{F}$ such that $F \subseteq S$.
    For every $S \in \mathcal{F}$, we construct the edge weighted graph $(G, w_S)$ by defining weights as $w_S(e) = 1$ if $e \in S$ and $w_S(e) = k + 1$ otherwise.
    Subsequently, we execute the polynomial-time algorithm that computes a minimum weighted edge cut $F^*$ of $(G, w_S)$.
    If $|F^*| \leq k$, then we output that $(G, \widehat{G}, k)$ is a {\yesinst}.
    Otherwise, we move to the next set from $\mathcal{F}$ and repeat the same procedure.
    If for every $S \in \mathcal{F}$, we construct the weighted graph $(G, w_S)$ and we see that the minimum weighted edge cut of $(G, w_S)$ is more than $k$, then we output that $(G, \widehat{G}, k)$ is a {\noinst}.
    Observe that the first step was to enumerate the family $\mathcal{F}$ of all maximal independent sets which is a polynomial-time procedure.
    After that, for each set $S$ in $\mathcal{F}$, it involves another polynomial-time procedure to check whether there is $F \subseteq S$ of size at most $k$ such that $F$ is a conflict-free cut or not.
    Therefore, this algorithm runs in polynomial-time.

Now we prove the correctness of the algorithm by considering both directions of the proof.
For the forward direction ($\Rightarrow$), assume $(G, \widehat{G})$ is a {\yesinst}. Then there exists $F \subseteq E$ such that $G-F$ is disconnected and $F$ is independent in $\widehat{G}$.
It is clear that $F$ is a subset of a maximal independent set $S \in \mathcal{F}$.
When the algorithm iterates over every maximal independent set in $\mathcal{F}$, it follows from \Cref{lemma:k-ind-covering} that the corresponding edge weighted graph $(G, w_S)$ must have an edge cut of total weight at most $k$.
 Hence, the algorithm for this particular $S$ outputs a {\yesinst}.
 
For the backward direction ($\Leftarrow$), assume the algorithm outputs a {\yesinst}.
This implies that for some $S \in \mathcal{F}$, the weighted graph $(G, w_S)$ has an edge cut of total weight at most $k$.
In such a case, such an edge cut $F$ satisfies that $F \subseteq S$ and for every $e \in F$, $w_S(e) =1$.
Hence, $|F| \leq k$.
Therefore, $(G, \widehat{G}, k)$ is a {\yesinst}.
\end{proof}

\section{Hardness Results}
\label{sec:hardness-res}

In this section, we provide some results that prove the non-existence of FPT algorithms under various parameterizations.
We divide this section into two subsections.
In \Cref{sec:hardness-structural}, we provide a dichotomy result on the parameterized tractability status with respect to structural parameters of the input graph (\Cref{sec:hardness-structural}).
Subsequently, in \Cref{sec:hardness-sol-size}, we prove W-hardness when conflict graph is proper interval graphs.

\subsection{Structural Parameters}
\label{sec:hardness-structural}

First, we give the proof that in class of all connected graphs with dissociation number one, {\CFCut} is polynomial-time solvable.

\begin{lemma}
    \label{lem:diss-1-polytime}

    If a conflict-free cut exists when the primary graph is in a class of connected graphs with dissociation number one, then there exists a conflict-free cut with at most two edges and it can be computed in polynomial time.
\end{lemma}

\begin{proof}
    Let $(G, \widehat{G})$ be an instance of {\CFCut} where the dissociation number (${\sf diss}$) of $G$ is one. Since $G$ is connected, $\delta(G) \geq 1$. Now if there exists a degree 1 vertex in $G$, then we can find a conflict-free cut of size 1 in polynomial time by cutting the edge incident to that degree 1 vertex. 
    Therefore, we can assume that $\delta(G) \geq 2$. Now since the dissociation number of $G$ is one, there exists a vertex $x \in V(G)$ such that $G - x$ is a disjoint union of edges.
    We can find such a vertex in polynomial time by checking for each vertex whether the graph obtained by removing that vertex is a disjoint union of edges. 

    \begin{claim}
    \label{claim:deg-2}
    For every $v \in V(G) \setminus \{x\}$, \ $d_G(v)=2 $ 
    \end{claim}

    \begin{claimproof}
    Assume for the sake of contradiction that there exists a vertex $v \in V(G) \setminus \{x\}$ such that $d_G(v) > 2$. Then there exists $u, w \in N_G(v) \setminus \{x\}$ such that $uv, wv \in E(G-x)$. This contradicts the fact that $G-x$ is a disjoint union of edges. Therefore, $d_G(v) \leq 2$ for every $v \in V(G) \setminus \{x\}$. Hence, the claim follows by the fact that $\delta(G) \geq 2$.
    \end{claimproof}

    \noindent
    Therefore, as a consequence of \Cref{claim:deg-2}, the structure of $G$ is as follows: \\
    Each edge in $G-x$ forms a triangle with $x$ in $G$ by adding the edges between $x$ and the endpoints of that edge. 

    \subparagraph*{Algorithm.} For each triangle in $G$ formed by $x$ and an edge in $G-x$, we check whether the edges of that triangle are pairwise adjacent in $\widehat{G}$. If there exists a triangle, with two non-adjacent edges in $\widehat{G}$, then we return those two edges as a conflict-free cut. Otherwise, we return that there is no conflict-free cut in $(G, \widehat{G})$. 
    
    \begin{claim}
    \label{claim:diss-1-correctness}
    $(G, \widehat{G})$ has a conflict-free cut of size 2 if and only if there exists a triangle in $G$ such that at least two of its edges are non-adjacent in $\widehat{G}$.
    \end{claim}

    \begin{claimproof}
        First, we give the forward direction $(\Rightarrow)$ of the proof. Suppose $F$ is a conflict-free cut of size 2 in $(G, \widehat{G})$. Now if $F$ contains two edges from different triangles in $G$, then $G-F$ remains connected as each triangle contains $x$ and $x$ is adjacent to all other vertices in $G$. This contradicts the fact that $F$ is an edge cut. Hence, $F$ contains two edges from the same triangle. Since $F$ is independent in $\widehat{G}$, those two edges are non-adjacent in $\widehat{G}$. 

        For the backward direction $(\Leftarrow)$, suppose there exists a triangle in $G$ such that at least two of its edges are non-adjacent in $\widehat{G}$. Let that triangle be formed by vertices $x, u, v \in V(G)$ and edges $xu, xv, uv \in E(G)$. If $xu$ and $xv$ are non-adjacent in $\widehat{G}$, then $F=\{xu, xv\}$ is a conflict-free cut as it separates the edge $uv$ from the graph. If $xu$ and $uv$ are non-adjacent in $\widehat{G}$, then $F=\{xu, uv\}$ is a conflict-free cut as it separates the vertex $u$ from the graph. If $xv$ and $uv$ are non-adjacent in $\widehat{G}$, then $F=\{xv, uv\}$ is a conflict-free cut as it separates the vertex $v$ from the graph. In all three cases, we obtain a conflict-free cut of size 2 in $(G, \widehat{G})$.
        
    \end{claimproof}

    \noindent
    \Cref{claim:diss-1-correctness} ensures the correctness of the algorithm, and clearly it runs in polynomial time as we only need to check for each triangle in $G$ whether at least two of its edges are non-adjacent in $\widehat{G}$.
    It completes the proof of lemma.
\end{proof}

{\DissNPhard*}

\begin{proof}
    The proof is by a reduction from {\ThreeColor} that is we state below.
    
    \defproblem{{\ThreeColor}}{An undirected graph $H = (V, E)$.}{Does there exist a function $c: V(H) \rightarrow [3]$ such that for every edge $uv \in E(H)$, $c(u) \neq c(v)$?}
     

        	\subparagraph*{Construction.} 
    Now we construct a graph $G$ as follows:
    \begin{enumerate}[(i)]
        \item We create two vertices $s$ and $t$, and add them to $G$.
        \item For every $v \in V(H)$, we create an $s$-$t$ path $P_{v}$ with vertices $\alpha_v, \beta_v$, and edges $s\alpha_v, \alpha_v \beta_v, \beta_v t$.
    \end{enumerate}
    
    \noindent
    Note that $V(G) = \{\alpha_v, \beta_v \mid v \in V(H)\} \cup \{s, t\}$ and $E(G) = \{s\alpha_v, \alpha_v \beta_v, \beta_v t \mid v \in E(H)\}$.
    What is crucial here is that $G - \{s,t\}$ is a disjoint union of edges.
    Therefore, $G$ has a dissociation set of size at most two.
    It remains to construct the conflict graph $\widehat{G}$ such that $V(\widehat{G}) = E(G)$.
    We define the edge set $E(\widehat{G})$ as follows:
    \begin{enumerate}[(i)]
        \item For every path $P_{v}$, the edges of $P_v$ forms a triangle in $\widehat{G}$, i.e. $s\alpha_v, \alpha_v \beta_v$, and $\beta_v t$ are pairwise adjacent in $\widehat{G}$.
        \item For every $xy \in E(H)$, we place conflicts between the corresponding edges in $G$, i.e.,
      $\{s\alpha_x, s\alpha_y\}, \{\alpha_x \beta_x, \alpha_y \beta_y\}, \{\beta_x t, \beta_y t\} \in E(\widehat{G})$.
    \end{enumerate}
    
    \noindent
    It remains to prove that $H$ is 3-colorable if and only if there exists a conflict-free cut in $(G,\widehat{G})$.
    First, we prove a structural claim about every conflict-free cut of $(G, \widehat{G})$.
    
    \begin{claim}
    \label{claim:color-reduction-cfc}
     If $F \subseteq E(G)$ is a conflict-free cut to $(G, \widehat{G})$, then $F$ separates $s$ and $t$.    
    \end{claim}
    
    \begin{claimproof}
    Assume for the sake of contradiction that $F$ does not separate $s$ and $t$.
    This implies that both $s$ and $t$ are in the same connected component of $G-F$.
    Since $F$ is a cut set of $G$, then without loss of generality, $F$ must separate some vertex other than $t$ from $s$ in $G$.
    Such a vertex is $\alpha_w$ or $\beta_w$ from $s$ for some $w \in V(H)$.
    For that to happen, $F$ must contain two edges from $\{s\alpha_w, \alpha_w \beta_w, \beta_w t\}$.
    But those edges are in pairwise conflict with each other.
    This leads to a contradiction.
    Thus, $F$ separates $s$ and $t$.
    \end{claimproof}

    \begin{claim}
    \label{claim:3-color-CFC-eqv}
    $H$ is 3-colorable if and only if $(G, \widehat{G})$ contains a conflict-free cut.
    \end{claim}
	
	\begin{claimproof}
	First, we give the backward direction $(\Leftarrow)$ of the proof. 
    Suppose that there exists a conflict-free cut $F$ in $(G,\widehat{G})$.
    Then $G-F$ is disconnected and $F$ is independent in $\widehat{G}$.
    It follows from \Cref{claim:color-reduction-cfc}, $F$ is a separator between $s$ from $t$.
    Hence, $F$ must contain at least one edge from each path $P_{v}$.
    Now as $F$ is independent in $\widehat{G}$, it can contain at most one edge from each path $P_{v}$.
    Hence, $F$ contains exactly one edge from each path $P_{v}$.

    Now we define a function $g : V(H) \to \{1,2,3\}$ as follows.
    \[
    	g(v) =
    	\begin{cases}
    		1 & \text{ if }s\alpha_v \in F\\
    		2 & \text{ if }\alpha_v \beta_v \in F\\
    		3 & \text{ otherwise}
    	\end{cases}
    \]
    We show that $g$ is a proper coloring of $H$.
    Assume for the sake of contradiction that this is not the case.
    Then, there is $xy \in E(H)$ such that $g(x) = g(y)$.
    This implies that $F$ contains one of these combinations of edges.
    If $g(x) = g(y) = 1$, then $s\alpha_x, s\alpha_y \in F$.
    Then, $s\alpha_x$ and $s\alpha_y$ are adjacent in $\widehat{G}$, leading to a contradiction that $F$ is independent in $\widehat{G}$.
    If $g(x) = g(y) = 2$, then $\alpha_x \beta_x, \alpha_y \beta_y \in F$.
    Then, then $\alpha_x \beta_x$ and  $\alpha_y \beta_y$ are adjacent in $\widehat{G}$, leading to a contradiction that $F$ is independent in $\widehat{G}$.
    If $g(x) = g(y) = 3$, then $\alpha_x t, \alpha_y t \in F$.
    Then $\alpha_x t$, and $\alpha_y t$ are adjacent in $\widehat{G}$, leading to a contradiction that $F$ is independent in $\widehat{G}$.
    Thus, $g$ is a proper 3-coloring of $H$.
	
    For the forward direction $(\Rightarrow)$,
    suppose that $H$ is 3-colorable.
    Then there exists a function $c : V(H) \to \{1,2,3\}$ 
    such that for every edge $uv \in E(H)$, we have $c(u) \neq c(v)$.\\
    We construct $F$ as follows.
    If $c(u) = 1$, then we put $s\alpha_u$ into $F$.
    If $c(u) = 2$, then we add $\alpha_u \beta_u$ into $F$.
    If $c(u) = 3$, then we add $\beta_u t$ into $F$.
    We claim that $F$ is a conflict-free cut in $(G,\widehat{G})$.\\
    As $c$ is a well-defined function, therefore, $F$ contains exactly one edge from each path $P_v$.
    Hence, $F$ is a cut set of $G$.\\
    Now we show that $F$ is independent in $\widehat{G}$.
    Assume for the sake of contradiction that $F$ is not independent.
    Then, it must be that there are $x, y \in V(H)$ such that $s\alpha_x, s\alpha_y \in F$ or $\alpha_x \beta_x, \alpha_y \beta_y \in F$ or $\beta_x t, \beta_y t \in F$. 
    Note that the choice of the edge in $F$ was based on $c(x)$ and $c(y)$.
    Additionally, each of the above three cases arise when $x$ and $y$ are adjacent in $H$.
    Observe that each of the three cases arise only when $c(x) = c(y)$.
    This leads to a contradiction that $c$ is a proper-coloring.
    Therefore, $F$ is independent in $\widehat{G}$.
    \end{claimproof}
    
    Observe that the above claims provide a polynomial-time reduction from \ThreeColor\ to \CFCut\ when the dissociation number of the input graph is at most 2.
    Therefore, {\CFCut} is NP-Complete even when the dissociation number of the input graph is a fixed constant.
    Hence, {\CFCut} is para-NP-hard when parameterized by ${\sf diss}$, the dissociation number of the input graph.
    This completes the proof of the first claim.
    
    Subsequently, it follows from \Cref{lem:diss-1-polytime} that {\CFCut} is polynomial-time solvable when the dissociation number is one.
    Therefore, our second claim also holds true and this completes the proof of our dichotomy result.
\end{proof}

\subsection{Parameterized by Solution Size}
\label{sec:hardness-sol-size}

\subparagraph*{Hardness when Conflict Graph is Unit Interval.}
Our next result is based on a W-hardness reduction from {\UnitMIS}.
We formally state the definition of {\UnitMIS} as follows.

\defparproblem{{\UnitMIS}}{Two unit-interval graphs $H_1 = (V,E_1)$ and $H_2 = (V,E_2)$, and a partition $V_1, V_2, \ldots, V_k$ of $V$.}{$k$}{Is there a set $S \subseteq V$, such that $S$ is an independent set in both $H_1$ and $H_2$, and for each $i \in [k]$, we have $|S \cap V_i| = 1$?}

We have the following proposition for {\UnitMIS}.

\begin{proposition}[Theorem 4 from \cite{AgrawalJKS20}]
\label{prop:hardness-unit-MIS}
{\UnitMIS} is W[1]-hard when parameterized by $k$.
\end{proposition}

Now, we describe the proof of our result.

{\UnitIntervalHard*}


\begin{proof}
We give a reduction from {\UnitMIS} problem defined as follows:
   
\subparagraph*{Construction Procedure.}
    Given an instance $(H_1, H_2, k, V_1, \ldots, V_k)$ of {\UnitMIS}, we construct an instance $(G, \widehat{G}, k')$ of {\mincut} as follows:
	%
	%
    Our construction procedure has two parts, first procedure constructs the primary graph, and the second procedure constructs the conflict graph.
    
    We construct the primary graph $G$ as follows:
    \begin{enumerate}[(i)]
        \item We create two designated terminal vertices $s$ and $t$.
        \item For every $v \in V(H)$, we create four vertices $v_1, v_2, v_3$ and $v_4$. We add the edges $v_1 v_2$, $v_1 v_3$ and $v_1 v_4$ to $G$.
        Add a $3k$-uncuttable gadget both between $v_2$ and $v_4$, and between $v_3$ to $v_4$.
        We call the subgraph induced by $\{v_1, v_2, v_3, v_4\}$ together with the gadgets, one between $v_2$ and $v_4$ and the other between $v_3$ and $v_4$ as a {\em rhombus} corresponding to $v$.
        	We consider $v_1$ as the {\em first vertex} of the rhombus corresponding $v$ and $v_4$ as the {\em last vertex} of the rhombus corresponding to $v$. 
        	
        \item For each partition of $V$, we create a series of rhombi.
        Specifically, for $V_i = \{v_{i,1}, v_{i,2}, \ldots, v_{i,n_i}\}$ where $i \in [k]$, we construct a rhombi $R_{i,1}, R_{i,2},  \ldots, R_{i,n_i}$ as in step $(ii)$.
        	We put the rhombi corresponding to the vertices of $V_i$ in the order $v_{i, 1},\ldots,v_{i, n_i}$ in such a way that the last vertex of $R_{i,j}$ with the first vertex of $R_{i,j+1}$ for every $1 \leq j \leq n_i - 1$.
        \item Finally, for every $i \in [k]$, we add a $3k$-uncuttable gadget between $s$ and the first vertex of $R_{i, 1}$ and add a $3k$-uncuttable gadget between the last vertex of $R_{i, n_i}$ with $t$.
    \end{enumerate}

    \noindent
    Observe that $G$ is connected.
    Now, we describe the construction procedure of the conflict graph $\widehat{G}$ such that $V(\widehat{G}) = E(G)$.
    We define the edge set $E(\widehat{G})$ as follows:
    \begin{enumerate}[(i)]
        \item For each $3k$-uncuttable gadget, the two edges of each length $2$ path are adjacent in $\widehat{G}$.
        Thus each gadget contributes a disjoint union of edges in $\widehat{G}$.
        \item The remaining vertices of $\widehat{G}$ correspond to the three edges $v_1 v_2, v_1 v_3, v_1 v_4$, of each rhombus corresponding to $v$ in $G$.
        In particular, for each $v \in V$, let
        $e_v^{2} := v_1v_2$, $e_v^{3} := v_1v_3$ and $e_v^{4} := v_1v_4$ denote these three (non-gadget) edges of the rhombus corresponding to $v$ in graph $G$.
        
        \item Now the collection $\{e_v^{2} \mid v \in V\}$ forms a graph isomorphic to $H_1$ and the collection $\{e_v^{3} \mid v \in V\}$ forms a graph isomorphic to $H_2$.
        	Equivalently, if $x, y \in V$ are adjacent in $H_1$, then we make $e_x^2$ and $e_y^2$ adjacent in $\widehat{G}$.
        	Similarly, if $x, y \in V$ are adjacent in $H_2$, then we make $e_x^3$ and $e_y^3$ adjacent in $\widehat{G}$. 
        	 
        \item Finally, for each $i \in [k]$, the edges $\{ e_v^{4} \mid v \in V_i\}$ forms a clique in $\widehat{G}$.
    \end{enumerate}
    
    Finally, we set $k' = 3k$. This completes the description of our reduction procedure and our {\mincfcut} instance is $(G, \widehat{G}, k')$.
    First, our next claim ensures that $\widehat{G}$ is a unit interval graph.

\begin{claim}
\label{claim:proper-interval-conflict}
    $\widehat{G}$ is a unit interval graph.
\end{claim}

\begin{claimproof}
    Consider a $3k$-uncuttable gadget.
    By construction, every $3k$-uncuttable gadget has $3k + 1$ edge-disjoint paths each of length exactly 2.
    And both the edges of each such length 2 path are adjacent in $\widehat{G}$.
    These vertices forms a disjoint union of edges in $\widehat{G}$.
    The remaining edges of $\widehat{G}$ are the edges of $H_1$, $H_2$ and the cliques corresponding to each partition $V_i$.
    Also, observe that the vertices of $H_1$, $H_2$, the cliques with vertices $e_v^4$ and the edges for the $3k$-uncuttable gadget are pairwise disjoint.
    That is, $\widehat{G}$ is isomorphic to the disjoint union of $H_1$, $H_2$, and a cluster graph $H^*$.
    Every component of this cluster graph $H^*$ is a clique based on the vertices $e_v^4 \in V_i$ or an edge whose endpoints are the two edges of a length 2 path of a $3k$-uncuttable gadget.
    As cluster graphs are unit interval graphs, and since unit interval graphs are closed under disjoint union, we conclude that $\widehat{G}$ is a unit interval graph.
\end{claimproof}

Before we prove the equivalence between two instances, we first justify some structural properties of our solution if $(G, \widehat{G}, k')$ is a {\yesinst}.

\begin{claim}
\label{claim:G-G-hat-property}
If $(G, \widehat{G}, k')$ is a {\yesinst}, then there exists a conflict-free edge cut $F$ of $G$ that 
\begin{itemize}
	\item disconnects $s$ from $t$, and
	\item for every $i \in [k]$, there exists a unique $u \in V_i$ such that $e_u^2, e_u^3, e_u^4 \in F$.
\end{itemize}
\end{claim}

\begin{claimproof}
Suppose that $(G, \widehat{G}, k')$ is a {\yesinst}. 
Then there exists a set $F \subseteq E(G)$ of size at most $3k$ such that $G-F$ is disconnected and $F$ is independent in $\widehat{G}$.
Assume that $(A, B)$ is the cut of $G$ such that $F = E_G(A, B)$.
We prove the items in the given order.

For the first item, suppose for the sake of contradiction, that $F$ does not separate $s$ and $t$ in $G$. 
This implies that both $s$ and $t$ belong to the same connected component of $G-F$. 
Since $(A, B)$ is a cut of $G$, there exists $x_i \in V(G) \setminus \{s, t\}$ such that $x_i \in B$ for some $i \in \{1,2,3,4\} $.
Note that there are $3k$-uncuttable gadgets between $x_2$ and $x_4$ and between $x_3$ and $x_4$.
In addition, if $x$ is the vertex of $H_1$ corresponding to the first rhombus, then there is a $3k$-uncuttable gadget between $s$ and $x_1$.
Hence, it must be that $x_2, x_3, x_4 \in B$ but $x_1 \in A$.
Assume that $x \in V_i$.
But additionally, there is a different vertex $y \in V_i$ such that $y_2, y_3, y_4 \in A$ and $y_1 \in B$.
Then, $x_1 x_4, y_1 y_4 \in F$.
But note that both $x$ and $y$ are in $V_i$ and $x_1 x_4 = e_x^4, y_1 y_4 = e_y^4$ and $\{e_v^4 \mid v \in V_i\}$ forms a clique in $\widehat{G}$.
This leads to a contradiction that both $e_x^4, e_y^4 \in F$.
Therefore, $s$ and $t$ must be separated by $F$.

Now, we consider the second item.
Due to the first item, $s$ is separated from $t$ by a conflict-free edge cut $F$.
Now, observe that $|F| \leq 3k$.
As $s \in A$, $t \in B$, and for every other vertex $w$, there is a $3k$-uncuttable gadget between $w_2$ and $w_4$ as well as between $w_3$ and $w_4$.
Hence, either $w_2, w_3, w_4 \in A$ $A$ or all in $w_2, w_3, w_4 \in B$.
In addition, for every $i \in [k]$, there is a vertex $u \in V_i$ such that $u_1 \in A$ and $u_2, u_3, u_4 \in B$.
These facts imply that $e_u^2, e_u^3, e_u^4 \in F$.
As $|F| \leq 3k$, there can exist only one $u \in V_i$ for every $i \in [k]$ such that $e_u^2, e_u^3, e_u^4 \in F$.
Hence, the second item follows.
\end{claimproof}

\begin{claim}
    \label{claim:unitMIS-cfc-eqv}
    $(H_1, H_2, k, V_1, \ldots, V_k)$ is a {\yesinst} if and only if $(G, \widehat{G}, k')$ is a {\yesinst}.
\end{claim}

\begin{claimproof}
First, we give the forward direction $(\Rightarrow)$ of the proof. 
Suppose that $(H_1, H_2, k, V_1, \ldots, V_k)$ is a {\yesinst}. 
Then there exists a set $S \subseteq V$ independent in both $H_1$ and $H_2$ such that for each $i \in [k]$, $|S \cap V_i| = 1$. 
Let $S = \{v_{1}, v_{2}, \ldots, v_{k}\}$ where $v_i \in V_i$ for each $i \in [k]$.
Now consider the set $F = \{ e_{v}^{2}, e_{v}^{3}, e_{v}^{4} \mid v \in S \} \subseteq E(G)$. 
Observe that $|F| = 3k = k'$ and $F$ separates $s$ and $t$ in $G$. Now we show that $F$ is independent in $\widehat{G}$. 
By construction, there are no edges of the form $e_u^{j} e_v^{\ell}$ in $\widehat{G}$ for $j \neq \ell$. 
Now since $S$ is independent in both $H_1$ and $H_2$, by step (iii) of the construction of $\widehat{G}$, there are no edges of the form $e_u^{2}e_v^{2}$ or $e_u^{3}e_v^{3}$ in $\widehat{G}$ for any $u, v \in S$. 
Finally, since $|S \cap V_i| = 1$ for each $i \in [k]$, step (iv) of the construction of $\widehat{G}$ ensures there are no edges of the form $e_u^{4}e_v^{4}$ in $\widehat{G}$ for any $u, v \in S$. 
Hence, $F$ is independent in $\widehat{G}$ and thus $(G, \widehat{G}, k')$ is a {\yesinst}.

For the backward direction $(\Leftarrow)$, 
Suppose that $(G, \widehat{G}, k')$ is a yes-instance.
Then, due to Claim \ref{claim:G-G-hat-property}, $s$ and $t$ are separated and for every $i \in [k]$, there exists a unique $w \in V_i$ such that $e_w^2, e_w^3, e_w^4 \in F$.
Now, we consider the set $S = \{w \in V(H_1) \mid e_w^2, e_w^3, e_w^4 \in F\}$.
Observe that $|S \cap V_i| = 1$.
Consider two arbitrary vertices $u, v \in S$.
Observe that $e_u^2 e_v^2 \notin E(H_1)$ as $F$ is conflict-free.
Similarly, $e_u^3 e_v^3 \notin E(H_2)$ as $F$ is conflict-free.
Hence, $S$ forms an independent set both in $H_1$ as well as in $H_2$.
Additionally, since $S$ satisfies the property that $|S \cap V_i| = 1$, it follows that $S$ is an independent set both in $H_1$ and $H_2$ such that $S$ intersects every $V_i$ in exactly one vertex.
This completes the proof that $(H_1, H_2, k, V_1, \ldots, V_k)$ is a {\yesinst}.
%
\end{claimproof}

Due to \Cref{prop:hardness-unit-MIS}, {\UnitMIS} is W[1]-hard and our reduction procedure constructs $(G, \widehat{G}, k')$ from $(H_1, H_2, V_1,\ldots, V_k)$ in polynomial-time such that $k' = 3k$.
Therefore, this procedure is a parameterized reduction.
Hence, {\mincfcut} is W[1]-hard when the conflict graph is a unit interval graph.
\end{proof}


\subparagraph*{Hardness when Conflict Graph is Bipartite.}
Our next result is based on a W-hardness reduction from {\MBC} \cite{cygan2015parameterized}.
We formally state the definition of {\MBC} as follows.

\defparproblem{{\MBC}}{A bipartite graph $H=(A \cup B, E)$, an integer $k$, a partition of $A$ into $k$ sets $A_1, A_2, \ldots, A_k$ and a partition of $B$ into $k$ sets $B_1, B_2, \ldots, B_k$.}{$k$}{Does there exist a subgraph of $H$ isomorphic to the biclique $K_{k,k}$ with one vertex from each of the sets $A_i$ and $B_i$?}

{\BipartiteHardness*}

\begin{proof}
    We give a reduction from {\MBC} problem defined as follows:
    
    \subparagraph*{Construction Procedure.}
    Given an instance $(H, k, A_1, \ldots, A_k, B_1, \ldots, B_k)$ of {\MBC}, we construct an instance $(G, \widehat{G}, k')$ of {\mincut} as follows: 
Our construction procedure has two parts, first procedure constructs the primary graph, and the second procedure constructs the conflict graph.
    
    We construct the primary graph $G$ as follows:
    \begin{enumerate}[(i)]
        \item We create two designated terminal vertices $s$ and $t$.
        \item For every $v \in V(H)$, we create four vertices $v_1, v_2, v_3$ and $v_4$. We add the edges $v_1 v_2$ and $v_1 v_3$ to $G$.
        Add a $4k$-uncuttable gadget both between $v_2$ and $v_4$, and between $v_3$ to $v_4$.
        We call the subgraph induced by $\{v_1, v_2, v_3, v_4\}$ together with the gadgets, one between $v_2$ and $v_4$ and the other between $v_3$ and $v_4$ as a {\em rhombus} corresponding to $v$.
        	We consider $v_1$ as the {\em first vertex} of the rhombus corresponding $v$ and $v_4$ as the {\em last vertex} of the rhombus corresponding to $v$.
        \item For each partition of $A$ and $B$, we create a series of rhombi. Specifically, for $A_i = \{a_{i,1}, a_{i,2}, \ldots, a_{i,n_i}\}$ and $B_i = \{b_{i,1}, b_{i,2}, \ldots, b_{i,m_i}\}$ where $i \in [k]$, we construct rhombi $R_{i,1}, R_{i,2}, \ldots, R_{i,n_i}$ and $R'_{i,1}, R'_{i,2}, \ldots, R'_{i,m_i}$ as in step $(ii)$. 
        We put the rhombi corresponding to the vertices of $A_i$ in the order $a_{i, 1},\ldots,a_{i, n_i}$ in such a way that the last vertex of $R_{i,j}$ identifies with the first vertex of $R_{i,j+1}$ for every $1 \leq j \leq n_i - 1$. 
        Similarly, we put the rhombi corresponding to the vertices of $B_i$ in the order $b_{i, 1},\ldots,b_{i, m_i}$ in such a way that the last vertex of $R'_{i,j}$ identifies with the first vertex of $R'_{i,j+1}$ for every $1 \leq j \leq m_i - 1$.
        
        \item Finally, for every $i \in [k]$, we add a $4k$-uncuttable gadget between $s$ and the first vertex of $R_{i, 1}$ and $R'_{i, 1}$ and add a $4k$-uncuttable gadget between the last vertex of $R_{i, n_i}$ and $R'_{i, m_i}$ with $t$.
    \end{enumerate}

    \noindent
    Observe that $G$ is connected.
    Now, we describe the construction procedure of the conflict graph $\widehat{G}$ such that $V(\widehat{G}) = E(G)$.
    We define the edge set $E(\widehat{G})$ as follows:
    \begin{enumerate}[(i)]
        \item For each $4k$-uncuttable gadget, the two edges of each length $2$ path are adjacent in $\widehat{G}$.
        Thus each gadget contributes a disjoint union of edges in $\widehat{G}$.
        \item The remaining vertices of $\widehat{G}$ correspond to the two edges $v_1 v_2, v_1 v_3$, of each rhombus corresponding to $v$ in $G$.
        In particular, for each $v \in V$, let
        $e_v^{2} := v_1v_2$ and $e_v^{3} := v_1v_3$ denote these two (non-gadget) edges of the rhombus corresponding to $v$ in graph $G$.
        \item We define the bipartite complement of a bipartite graph $H = (A \cup B, E)$ as $\overline{H} = (A \cup B, \overline{E})$, where $\overline{E} = \{ uv \mid u \in A, v\in B, uv \notin E \}$. Now the collection $\{ e_v^2 \mid v \in V(H)\}$ forms a graph, say $\overline{H}_2$, isomorphic to $\overline{H}$ and the collection $\{ e_v^3 \mid v \in V(H)\}$ forms a graph, say $\overline{H}_3$, isomorphic to $\overline{H}$. 
        Equivalently, if $x \in A, y \in B$ are non-adjacent in $H$, then we make $e_x^2$ and $e_y^2$ adjacent in $\widehat{G}$ and $e_x^3$ and $e_y^3$ adjacent in $\widehat{G}$. 
         
        \item Finally, for each $i \in [k]$:
        \begin{enumerate}
            \item $\{ e_u^2, e_v^3 \} \in E(\widehat{G})$ for all $u,v \in A_i, u \neq v$.
            \item $\{ e_u^2, e_v^3 \} \in E(\widehat{G})$ for all $u,v \in B_i, u \neq v$.
        \end{enumerate}
        
    \end{enumerate}

     \noindent
    We set $k' = 4k$. This completes the description of our reduction procedure and our {\mincfcut} instance is $(G, \widehat{G}, k')$.
    First, our next claim ensures that $\widehat{G}$ is a bipartite graph.
    
\begin{claim}
\label{claim:bipartite-conflict}
    $\widehat{G}$ is a bipartite graph.
\end{claim}


\begin{claimproof}
By definition of bipartite complement, $\overline{H}_2$ and $\overline{H}_3$ are bipartite graphs.
Let the partition of $\overline{H}_2$ be $(\widehat{A_2}, \widehat{B_2})$ and the partition of $\overline{H}_3$ be $(\widehat{A_3}, \widehat{B_3})$.
Now consider the partition $(\widehat{A_2} \cup \widehat{B_3}, \ \widehat{A_3} \cup \widehat{B_2})$ of vertices of $\widehat{G}[V(\overline{H}_2) \cup V(\overline{H}_3)]$.
Since $\overline{H}_2$ and $\overline{H}_3$ are bipartite, each of $\widehat{A_2}, \widehat{A_3}, \widehat{B_2}$ and $\widehat{B_3}$ are independent sets.
Now by the construction of $\widehat{G}$, vertices of $\widehat{A_2}$ are adjacent to vertices of $\widehat{A_3}$ and $E(\widehat{A_2}, \widehat{B_3}) = \emptyset$. Similarly, vertices of $\widehat{B_2}$ are adjacent to vertices of $\widehat{B_3}$ and $E(\widehat{B_2}, \widehat{A_3}) = \emptyset$. 
Hence, $\widehat{G}[V(\overline{H}_2) \cup V(\overline{H}_3)]$ is a bipartite graph with this partition. 
Also, the edges of each gadget form a $1$-regular graph in $\widehat{G}$. Since the disjoint union of bipartite graphs is bipartite, and $\widehat{G}$ is the edge-disjoint union of the graphs induced by the gadgets and $\widehat{G}[V(\overline{H}_2) \cup V(\overline{H}_3)]$, we conclude that $\widehat{G}$ is bipartite.  
\end{claimproof}

Before we prove the equivalence between two instances, we first justify some structural properties of our solution if $(G, \widehat{G}, k')$ is a {\yesinst}.

\begin{claim}
\label{claim:G-G-hat-property-bip}
If $(G, \widehat{G}, k')$ is a {\yesinst}, then every conflict-free edge cut $F$ of $G$ of size at most $4k$ satisfies the following properties:
\begin{itemize}
	\item disconnects $s$ from $t$, and
	\item for every $i \in [k]$, there exists a unique $u \in A_i$  and a unique $v \in B_i$ such that $e_u^2, e_u^3, e_v^2, e_v^3 \in F$.
\end{itemize}
\end{claim}

\begin{claimproof}
Suppose that $(G, \widehat{G}, k')$ is a {\yesinst}. 
Then there exists a set $F \subseteq E(G)$ of size at most $4k$ such that $G-F$ is disconnected and $F$ is independent in $\widehat{G}$.
Assume that $(L, R)$ is the cut of $G$ such that $F = E_G(L, R)$. 
We prove the items in the given order.


For the first item, suppose for the sake of contradiction, that $F$ does not separate $s$ and $t$ in $G$. 
This implies that both $s$ and $t$ belong to the same connected component of $G-F$. 
Without loss of generality, assume that $s,t \in L$. Since $(L, R)$ is a cut of $G$, there exists $x \in V(G) \setminus \{s, t\}$ such that $x \in R$. 
Without loss of generality, assume that $x$ is a vertex of the series of rhombi corresponding to the partition $A_1$ and for simplicity, let $A_1 = \{ a_1, a_2, \ldots, a_{n_1} \}$. Let $R_1, R_2, \ldots, R_{n_1}$ denote the corresponding series of rhombi in the primary graph $G$. 
Let $a_j \in A_1$ be the vertex corresponding to the rhombus $R_j$ such that $x$ is a vertex of $R_j$. 
Now if $x$ is the first vertex of $R_1$, then there is a $4k$-uncuttable gadget between $s$ and $x$. Therefore, $x$ cannot be the first vertex of $R_1$. Similarly, since the fourth vertex of $R_{n_1}$ has a $4k$-uncuttable gadget with the second vertex of $R_{n_1}$, the third vertex of $R_{n_1}$ and $t$; hence $x$ cannot be the second, third or fourth vertex of $R_{n_1}$. 
Now if $x$ is a first vertex of any rhombus other than $R_1$, then in order to separate $s$ and $x$ without cutting any $4k$-uncuttable gadgets, $F$ must contain $e_{a_i}^2$ and $e_{a_i}^3$ for some $a_i \in A_1$ such that $i < j$. 
Similarly, in order to separate $t$ and $x$ without cutting any $4k$-uncuttable gadgets, $F$ must contain $e_{a_i}^2$ and $e_{a_i}^3$ for some $a_i \in A_1$ such that $i \geq j$.
Now if $x$ is a second, third or fourth vertex of any rhombus other than $R_{n_1}$, then in order to separate $s$ and $x$ without cutting any $4k$-uncuttable gadgets, $F$ must contain $e_{a_i}^2$ and $e_{a_i}^3$ for some $a_i \in A_1$ such that $i \leq j$.
Similarly, in order to separate $t$ and $x$ without cutting any $4k$-uncuttable gadgets, $F$ must contain $e_{a_i}^2$ and $e_{a_i}^3$ for some $a_i \in A_1$ such that $i > j$.
In both cases, we have $F$ must contain $e_{a_i}^2$ and $e_{a_i}^3$ for some $a_i \in A_1$ and $e_{a_{\ell}}^2$ and $e_{a_{\ell}}^3$ for some $a_{\ell} \in A_1$ such that $i \neq \ell$. 
But then $\{e_{a_i}^2, e_{a_{\ell}}^3\} \in E(\widehat{G})$ and hence $F$ is not independent in $\widehat{G}$, which is a contradiction.

Now, we consider the second item.
Due to the first item, $s$ is separated from $t$ by a conflict-free edge cut $F$.
Now, observe that $|F| \leq 4k$.
As $s \in L$, $t \in R$, and for every other vertex $w$, there is a $4k$-uncuttable gadget between $w_2$ and $w_4$ as well as between $w_3$ and $w_4$.
Hence, either $w_2, w_3, w_4 \in L$  or all in $w_2, w_3, w_4 \in R$.
In addition, for every $i \in [k]$, there is a vertex $u \in A_i$ such that $u_1 \in L$ and $u_2, u_3, u_4 \in R$ and there is a vertex $v \in B_i$ such that $v_1 \in R$ and $v_2, v_3, v_4 \in L$.
These facts imply that $e_u^2, e_u^3, e_v^2, e_v^3 \in F$.
As $|F| \leq 4k$, there can exist only one $u \in A_i$ and only one $v \in B_i$ for every $i \in [k]$ such that $e_u^2, e_u^3, e_v^2, e_v^3 \in F$.
Hence, the second item follows.
\end{claimproof}

\begin{claim}
\label{claim:mbc-cfc-eqv}
    $(H, k, A_1, \ldots, A_k, B_1, \ldots, B_k)$ is a {\yesinst} if and only if $(G, \widehat{G}, k')$ is a {\yesinst}.
\end{claim}

\begin{claimproof}
    First, we give the forward direction $(\Rightarrow)$ of the proof. Suppose $H$ has a multicolored biclique. Then there exists a subgraph of $H$ isomorphic to the biclique $K_{k,k}$ with one vertex from each of the sets $A_i$ and $B_i$. Let that subgraph be induced by the vertices $\{a_1, a_2, \ldots, a_k\} \cup \{b_1, b_2, \ldots, b_k\}$ where $a_i \in A_i$ and $b_i \in B_i$. 
    Now we construct a conflict-free cut in $(G, \widehat{G})$ as follows. For each rhombus in $\widehat{G}$ corresponding to the vertex $a_i$, we select the edges $e_{a_i}^2$ and $e_{a_i}^3$ into the cut. Similarly, for each rhombus corresponding to the vertex $b_i$, we select the edges $e_{b_i}^2$ and $e_{b_i}^3$ into the cut. Therefore, $F = \{e_{a_i}^2, e_{a_i}^3, e_{b_i}^2, e_{b_i}^3 \mid  i \in [k]\} \subseteq E(G)$ is a cut of size exactly $4k$ as it separates $s$ from $t$.
    Now we show that $F$ is independent in $\widehat{G}$. Observe that $F \subseteq V(\overline{H}_2) \cup V(\overline{H}_3)$ and since $\widehat{G}[V(\overline{H}_2) \cup V(\overline{H}_3)]$ is a bipartite graph with the partition $(\widehat{A_2} \cup \widehat{B_3}, \ \widehat{A_3} \cup \widehat{B_2})$, the only possible edges among the vertices of $F$ in $\widehat{G}$ are: 
    \begin{enumerate}[(i)]
        \item $\{e_{a_i}^2, e_{a_j}^3\}$ for some $i,j \in [k]$.
        \item $\{e_{b_i}^2, e_{b_j}^3\}$ for some $i,j \in [k]$.
        \item $\{e_{a_i}^2, e_{b_j}^2\}$ for some $i,j \in [k]$.
        \item $\{e_{b_i}^3, e_{a_j}^3\}$ for some $i,j \in [k]$.
    \end{enumerate}
Now for case (i) and (ii), the step $(iv)$ of the construction of $\widehat{G}$ ensures that there are no such edges in $\widehat{G}$ as either $a_i$ and $a_j$ lies in different partitions of $A$ or $a_i = a_j$ (equivalently, $b_i$ and $b_j$ lies in different partitions of $B$ or $b_i = b_j$) and in both cases, $e_{a_i}^2$ and $e_{a_j}^3$ (equivalently, $e_{b_i}^2$ and $e_{b_j}^3$) are non-adjacent in $\widehat{G}$ for all $i,j \in [k]$. For case (iii) and (iv), the step $(iii)$ of the construction of $\widehat{G}$ ensures that there are no such edges in $\widehat{G}$ as $a_i b_j \in E(H)$ for all $i,j \in [k]$ and hence, $e_{a_i}^2$ and $e_{b_j}^2$ (equivalently, $e_{b_i}^3$ and $e_{a_j}^3$) are non-adjacent in $\widehat{G}$ for all $i,j \in [k]$. Therefore, $F$ is independent in $\widehat{G}$. Hence, $F$ is a conflict-free cut of size at most $4k$ in $(G, \widehat{G})$.
     
    For the backward direction $(\Leftarrow)$, suppose $(G, \widehat{G}, k')$ is a {\yesinst}. Then there exists a conflict-free cut $F$ of size at most $4k$ in $(G, \widehat{G})$. 
Then, due to \Cref{claim:G-G-hat-property-bip}, $s$ and $t$ are separated and for every $i \in [k]$, there exists a unique $u \in A_i$  and a unique $v \in B_i$ such that $e_u^2, e_u^3, e_v^2, e_v^3 \in F$.
 Now for each series of rhombi corresponding to the partition $A_i$, let the two edges in $F$ be $e_{a_i}^2$ and $e_{a_i}^3$ for some vertex $a_i \in A_i$. Similarly, for each series of rhombi corresponding to the partition $B_i$, let the two edges in $F$ be $e_{b_i}^2$ and $e_{b_i}^3$ for some vertex $b_i \in B_i$. Now we claim that $\{a_1, a_2, \ldots, a_k\} \cup \{b_1, b_2, \ldots, b_k\}$ induces a complete bipartite graph in $H$.
    Assume for the sake of contradiction that there exists $i,j \in [k]$ such that $a_i b_j \notin E(H)$. Then by the construction of $\widehat{G}$, we have $\{e_{a_i}^2, e_{b_j}^2\} \in E(\widehat{G})$ and $\{e_{b_j}^3, e_{a_i}^3\} \in E(\widehat{G})$. This contradicts the fact that $F$ is independent in $\widehat{G}$. 
    Since for every $i \in [k]$, $a_i \in A_i$ and $b_i \in B_i$, we conclude that $\{a_1, a_2, \ldots, a_k\} \cup \{b_1, b_2, \ldots, b_k\}$ induces a multicolored biclique in $H$.
\end{claimproof}

SInce {\MBC} is W[1]-hard and our reduction procedure constructs $(G, \widehat{G}, k')$ from $(H, k, A_1, \ldots, A_k, B_1, \ldots, B_k)$ in polynomial-time such that $k' = 4k$.
Therefore, this procedure is a parameterized reduction.
Hence, {\mincfcut} is W[1]-hard when the conflict graph is a bipartite graph.

\end{proof}

\section{Negative Kernelization Result}
\label{sec:neg-ker-res}

In this section, we complement our fixed-parameter tractability result of \Cref{sec:solution-size-results}.
Particularly, we provide that unless {\nka}, {\mincfcut} admits no polynomial kernel when parameterized by the smallest number of edges crossing the conflict-free edge cut even when the conflict graph is bounded degenerate.
Additionally, we also provide the nonexistence of polynomial kernel for {\mincfcut} when parameterized by the vertex cover number of the input graph.
Our result uses cross-composition framework as illustrated in \Cref{sec:prelims}.

\subsection{Parameterized by Solution Size when Conflict Graph is 1 regular}

We now present the cross-composition from {\cfcut} to {\mincfcut} parameterized by the sum of the solution size and the vertex integrity of the primary graph even when the conflict graph is 1-regular. 
This establishes the non-existence of a polynomial kernel for the latter problem, even under the structural restriction that the conflict graph is 1-regular. 



\begin{lemma}
\label{lemma:cross-compo-sol-size}
Given $2^t$ distinct instances $(G_1,\widehat{G_1}),\ldots,(G_{2^t},\widehat{G_{2^t}})$ of {\cfcut} such that for every $i$, $|V(G_i)|=n$, $G_i$ is 3-degenerate and $\widehat{G_i}$ is 1-regular, there exists a polynomial-time algorithm that constructs an instance $(H, \widehat{H})$ of {\cfcut} such that:
\begin{enumerate}[(i)]
    \item The primary graph $H$ is 3-degenerate and the conflict graph $\widehat{H}$ is 1-regular.
    \item $(H, \widehat{H})$ is a {\yesinst} if and only if there exists an $i \in [2^t]$ such that $(G_i, \widehat{G_i})$ is a {\yesinst}.
    \item Moreover, suppose that $(A, B)$ is a conflict-free edge cut of $H$ with the smallest possible number of edges.
    Then, there exists $i \in [2^t]$ such that the edge cut $E_H(A, B)$ is also a conflict-free edge cut of smallest cardinality of $G_i$.
    \item The vertex integrity of $H$ and the minimum number of edges of any conflict-free edge cut of $H$ if exists is at most polynomial in $n + t$.
\end{enumerate}
\end{lemma}


\begin{proof}
Suppose that there are $2^t$ distinct instances $(G_1,\widehat{G_1}),\ldots,(G_{2^t},\widehat{G_{2^t}})$ of {\cfcut} such that for every $i$, $|V(G_i)|=n$, $G_i$ is 3-degenerate and $\widehat{G_i}$ is 1-regular.
    We construct a single instance $(H, \widehat{H})$ of \CFCut\ as follows, where $H$ and $\widehat{H}$ is the primary graph and the conflict graph, respectively.
    The construction proceeds as follows:
    \begin{itemize}
        \item \textbf{Primary graph $H$:} Start with the disjoint union of the graphs $G_1, \dots, G_{2^t}$.
        As every $G_i$ is 3-degenerate, we consider the degeneracy sequence of $G_i$.
        For each $i \in [2^t]$, choose {the first vertex $w_i$ in the 3-degeneracy sequence of $G_i$}.
        Identify {the vertex subset $\{w_1,\ldots,w_{2^t}\}$} into a single vertex $w$.
        Formally,
        \[
            V(H) = \{w\} \cup \bigcup_{i\in [2^t]} \left( V(G_i) \setminus \{w_i\} \right),
        \qquad
            E(H) = \bigcup_{i\in [2^t]} E(G_i - w_i) \cup \{wx \mid xw_i \in E(G_i)\},
        \]
        \item \textbf{Conflict graph $\widehat{H}$:} The vertex set of $\widehat{H}$ is precisely $E(H)$.
        	{If $xw_i$ is adjacent to $uv$ in $\widehat{G}_i$, then make $xw$ and $uv$ adjacent in $\widehat{H}$.
        	Additionally, if $xy$ and $uv$ are adjacent in $\widehat{G_i}$ such that $w_i \neq x, y, u, v$, then $xy$ and $uv$ are adjacent in $\widehat{H}$.}
    \end{itemize}
    This construction is clearly computable in time polynomial in the sum of the size of all the instances.
    We prove the items in the given order.
%
%
%
    %
    
    {\bf Proof of (i):}
    Let us prove the first item that the primary graph $H$ is 3-degenerate and conflict graph $\widehat{H}$ is 1-regular.
    
    {First, we argue that $H$ is 3-degenerate.
    Note that every $G_i$ is 3-degenerate and $w_i$ is the first vertex of $G_i$ in the 3-degeneracy sequence.
    We consider the following sequence of $V(H)$ as follows.
    First, we put $w$ in the sequence.
    Next, we put the vertices of $G_1$ in the same order as the 3-degeneracy sequence of $G_1$.
    After that, for each $i = 2,3,\ldots,2^t$ in this order, we put the vertices of $G_i$ as in the same order as per the 3-degeneracy sequence of $G_i$.
    Observe that every vertex $u$ other than $w$ is present only in a unique $G_i$.
    Additionally, $u$ is only adjacent to at most 3 vertices (including $w$) that appears before $u$ in the sequence.
    As $w$ is the first vertex in the sequence, this sequence satisfies the property of being 3-degenerate.
    Therefore, $H$ is 3-degenerate.
    
    Now, we consider $\widehat{H}$.
    It is sufficient to argue that any edge of $H$ is adjacent to exactly one other edge of $H$.
    First, we consider the case for an edge $uv \in E(H)$ such that $w \neq u, v$.
    Note that $uv$ is present in a unique $G_i$.
    Observe that in the graph $\widehat{G_i}$, $uv$ is adjacent to only one edge of $G_i$, which is $w_ix$ or some other edge for which $w_i$ is not an endpoint.
    If $uv$ is adjacent to $w_ix$ in $\widehat{G_i}$, then $uv$ is adjacent to only $wx$ in $\widehat{H}$.
    Otherwise, $uv$ is adjacent to the unique edge of $H$ in which $w_i$ is not an endpoint.
    Now, we consider an edge $xw \in E(H)$.
    Note that this edge is represent by $xw_i \in E(G_i)$ for a unique $G_i$, and in the graph $\widehat{G_i}$, $xw_i$ is adjacent to a unique edge ($yw_i$ or an edge for which $w_i$ is not an endpoint).
    Then, $xw$ is also adjacent to a unique edge that is $yw$ or an edge for which $w_i$ is not an endpoint.
    Hence, every edge of $H$ is adjacent to a unique vertex in $\widehat{H}$.
    Therefore, $\widehat{H}$ is 1-regular.}

    {\bf Proof of (ii):}
    Now, we prove the second item. 
    We prove the two implications; forward direction ($\Rightarrow$) and backward direction ($\Leftarrow$) separately.
        For the sake of forward $(\Rightarrow)$ direction proof, suppose that $(H, \widehat{H})$ is a {\yesinst} and let $F \subseteq E(H)$ be a conflict-free cut with the smallest number of edges.
        It implies that there exists a nontrivial bipartition $(A, B)$ of $V(H)$ such that $E_H(A, B)$ is independent in $\widehat{H}$.
        For each $i \in [2^t]$, let $A_i = A \cap V(G_i)$ and $B_i = B \cap V(G_i)$.
        Note that $(A_i, B_i)$ is precisely $F \cap E(G_i)$ is a valid edge cut of $G_i$ when both $A_i, B_i \neq \emptyset$.
        Suppose that $F_i = F\cup E(G_i)$.
        Since $F$ is independent in $\widehat{H}$ and conflicts are only contained within each individual $\widehat{G_i}$, it follows that each $F_i$ is independent in $\widehat{G_i}$.

        It remains to show that at least one $F_i$ is a \emph{valid} edge cut of $G_i$, meaning that the bipartition is nontrivial.
        Suppose, for the sake of contradiction, that for every $i \in [2^t]$, $F_i = \emptyset$.
        Since each $G_i$ is connected, the absence of crossing edges implies that $V(G_i)$ is {\intra}; i.e., either $V(G_i) \subseteq A$ or $V(G_i) \subseteq B$.
        However, for all $G_i$'s when looked at in $H$ share the single vertex $w$ (which is the identification of the $\{w_i \mid i \in [2^t]\}$).
        Since $w$ belongs to a single side of the partition, say $w \in A$, it follows that for every $i$, the entire vertex set $V(G_i)$ is contained in $A$.
        Thus, $A = V(H)$ and $B = \emptyset$, contradicting the assumption that $E_H(A, B)$ is {nonempty}.
        Therefore, there exists an index $j \in [2^t]$ such that $F_j \neq \emptyset$.

        Since $F_j \neq \emptyset$, the bipartition $(A_j, B_j)$ of $G_j$ is nontrivial, so $F_j$ is an edge cut of $G_j$.
        Moreover, $|F_j| \le |F|$ and $F_j$ is independent in $\widehat{G_j}$.
        Hence, $(G_j, \widehat{G_j}, k)$ is a \yesinst\ of {\CFCut}.

        Now, for the sake of backward $(\Leftarrow)$ direction, suppose that there exists $j \in [2^t]$ such that $(G_j, \widehat{G_j})$ is a \yesinst.
        Then, there exists a bipartition $(A_j, B_j)$ of $V(G_j)$ such that $E_{G_j}(A_j, B_j) \neq \emptyset$ and is independent in $\widehat{G_j}$.
        We consider the $j \in [2^t]$ such that $E_{G_j}(A_j, B_j) \neq \emptyset$ is smallest among all possible conflict-free cuts in $G_1,\ldots,G_{2^t}$.
        Without loss of generality, assume $w_j \in A_j$.
        We extend this to a bipartition $(A, B)$ of $V(H)$:
        \begin{itemize}
            \item Assign $w$ to $A$.
            Subsequently, for all vertices $v \in V(G_j)$, assign $v$ to $A$ if and only if $v \in A_j$.
           	Otherwise, assign $w$ to $B$.
            \item For every $i \neq j$, since the entire subgraph $G_i$ is attached to the rest of $H$ only through the identified vertex $w$, assign every vertex of $V(G_i) \setminus \{w_i\}$ to $A$.
        \end{itemize}
        
        Now, consider the edge-cut $F = E_H(A, B)$.
        The only edges in $H$ that cross between $A$ and $B$ are exactly the edges the cut $(A_j, B_j)$, because all other $G_i$'s are placed entirely within side $A$.
        Hence, $F = F_j$, {which is a conflict-free cut of the smallest cardinality that $H$ can have}.
        As $F_j$ is independent in $\widehat{G_j}$, $F$ is independent in $\widehat{H}$.
        Therefore, $(A, B)$ is a valid conflict-free edge cut of $H$, and $(H, \widehat{H})$ is a \yesinst.
        This completes the proof of the claim.
    Now, we consider the third and fourth item.
    
    {\bf Proof of (iii):}
    Suppose that $(H, \widehat{H})$ is a {\yesinst} and $(A, B)$ denote a conflict-free edge cut with smallest number of edges crossing the cut.
    In the proof argument of item-(ii), we have already explained that there exists $i \in [2^k]$ such that $(G_i, \widehat{G_i})$ that has a nonempty conflict-free cut.
    We consider every $i \in [2^t]$ such that the cut $(A, B)$ restricted to $G_i$ which is $(A_i, B_i)$ provides a nonempty edge cut.
    Note that every such $(A_i, B_i)$ provides a nonempty conflict-free edge cut.
    Among all such restrictions to $(G_i, \widehat{G_i})$, we choose the one $j \in [2^t]$ such that the restriction $(A_j, B_j)$ provides the smallest number of edges in the corresponding edge-cut.
    Observe that keeping $A_j$ on one side and every other vertex of $H$ on the other side also provides a conflict-free cut of $H$ and it contains the same number of edges as $E_G(A_j, B_j)$.
    In case $E_G(A_j, B_j) \subsetneq E_H(A, B)$, then this contradicts our choice of $(A, B)$ as a conflict-free cut of $H$ with the smallest number of edges.
    This completes the proof of item-(iii).
    
    {\bf Proof of item-(iv):}
    It remains to argue that both the vertex integrity of $H$ and the minimum number of edges crossing any conflict-free cut of $H$ is indeed upper bounded by polynomial in $n + t$.
    Now, we consider the vertex integrity of $H$.
    Observe that $H$ has a cut vertex $w$ and every connected component of $H - \{w\}$ has at most $n-1$ vertices.
    Hence, the vertex integrity of $H$ is at most $n$ which is at most polynomial in the maximum possible size of all the instances and $t$.
    Additionally, the minimum number of edges crossing any conflict-free cut $H$ is also upper bounded by the maximum possible size of a conflict-free cut of among all the graphs $G_j$s.
    Hence, all our claims hold.
\end{proof}

Now, using the above result, we prove the following statement.

{\SolSizeNoKernelCFC*}


\begin{proof}
Note that {\CFCut} is NP-hard when the input graph is 3-degenerate and the conflict graph is 1-regular.
We prove this statement by providing a cross composition of {\CFCut} to {\mincfcut}.
An instance $(G, \widehat{G})$ of {\CFCut} is {\em malformed} if $E(G) \neq \widehat{G}$, and is {\em well-formed} otherwise.
We define an equivalence relation in which all malformed instances are equivalent.
Subsequently, two well-formed instances $(G_i, \widehat{G_i})$ and $(G_j, \widehat{G_j})$ are equivalent if $|V(G_i)| = |V(G_j)|$.
{Whether two well-formed} instances $(G_1, \widehat{G_1})$ and $(G_2, \widehat{G_2})$ of {\CFCut} belong to the same equivalence class or not can be decided in polynomial-time by verifying the conditions $|V(G_1)| = |V(G_2)|$.
{We argue why this is a polynomial equivalence relation.
    We can check in polynomial-time whether two given instances belong to one equivalence class or different equivalence classes.
    Note that all malformed instances form a single equivalence class.
    Two instances with the same number of vertices belong to one equivalence class.
    Hence, a finite set $S$ of instances with at most $n$ vertices form at most $\cOO(n)$ equivalence classes.}

Consider $2^t$ instances $(G_1, \widehat{G_1}), (G_2, \widehat{G_2}),\ldots,(G_{2^t}, \widehat{G_{2^t}})$ of {\CFCut} such that each $G_i$ has exactly $n$ vertices.
We invoke \Cref{lemma:cross-compo-sol-size} to produce $(H, \widehat{H})$ such that the claimed properties are satisfied.
Suppose that $k^*$ is the maximum number of edges among all graphs in $G_1,\ldots,G_{2^t}$.
 We set $(H, \widehat{H}, k^*)$ as our output instance of {\mincfcut}.
 
First of all observe from the item-(ii) of \Cref{lemma:cross-compo-sol-size} that if no $(G_i, \widehat{G_i})$ is a yes-instance to {\CFCut}, then there cannot exist any conflict-free cut of $H$ and vice versa.
Subsequently, from item-(iii) of \Cref{lemma:cross-compo-sol-size}, the conflict-free edge cut $E_H(A, B)$ of smallest cardinality in $H$ is also a conflict-free edge cut in $G_i$ for some $i \in [2^t]$.
Hence, the number of edges in $E_H(A, B)$ is at most the number of edges in $G_i$ which is also upper bounded by $k^*$.
Therefore, $|E_H(A, B)| \leq k^*$.
Hence, $(H, \widehat{H}, k^*)$ is a {\yesinst} to {\mincfcut}.
Due to item-(iv), $k^*$ is also upper bounded by polynomial in $n + t$ which is $\cOO(n^2)$ and the vertex integrity $H$ is also upper bounded by polynomial in $n + t$.
Therefore, this procedure provides us a valid cross composition from {\cfcut} to {\mincfcut}.
Due to item-(i), $H$ is 3-degenerate and $\widehat{H}$ is 1-regular.
Hence, {\mincfcut} parameterized by the sum of solution size $k$ and the vertex integrity of the primary graph admits no polynomial kernel unless {\nka}.
\end{proof}

\subsection{Parameterized by the Vertex Cover Number}

Given an undirected graph $G$, the {\em distance to clique} is the minimum number of vertices whose deletion results in a clique.
We use ${\sf dtc}$ to denote the distance to clique.
We provide this result by a polynomial parameter transformation from {\sc Independent Set} parameterized by the distance to clique (${\sf dtc}$) to {\mincut} parameterized by the vertex cover number (${\sf vc}$) as follows.
For the sake of completeness, we define {\sc Independent Set} below.

\defproblem{{\sc Independent Set}}{An undirected graph $H=(V_H, E_H)$, an integer $\ell$.}{Does $H$ have an independent set of size at least $\ell$?}

Due to Bodlaender et al. \cite{bodlaender_cross-composition_2010}, unless {\nka}, {\sc Independent Set} admits no polynomial kernel with respect to the the parameter distance to clique (${\sf dtc}$) of the input graph.

\begin{lemma}
\label{lemma:min-cf-cut-vc}
There is a polynomial parameter transformation from the {\sc Independent Set} parameterized by the distance to clique (${\sf dtc}$) to the {\mincfcut} parameterized by the vertex cover number (${\sf vc}$) of the primary graph.
In particular, this polynomial parameter transformation also satisfies that the vertex cover number (${\sf vc}$) is $\cOO({\sf dtc})$.
\end{lemma}

\begin{proof}
We provide a parameter preserving reduction from {\sc Independent Set} (with respect to the parameter ${\sf dtc}$, {distance to clique}) to {\mincfcut} (with respect to the parameter ${\sc vc}$).
 %

   \noindent
   Given an instance $(H, \ell)$ of {\sc Independent Set}, we first in polynomial-time compute a set $X$ of at most $2{\sf dtc}$ vertices such that $G - X$ is clique.
   From this instance, we construct an instance $(G, \widehat{G}, \ell')$ of {\mincfcut} such that $Y$ is vertex cover of $G$ of size $k' = 4{\sf dtc} + 4$.

    \subparagraph*{Construction.} 
    We construct the primary graph $G$ as follows:
        \begin{enumerate}[(i)]
        \item The vertex set of $G$ is as follows:
        $V(G) = V(H)  \cup \{s,t,t_1,t_2\} \cup \{a_x,b_x \mid x\in X \cup \{s \} \}$ 
        \item For every vertex $v \in V(H)$, we add three edges $sv,\ vt_1$ and $vt_2$.
        \item For every $x \in X \cup \{s\}$, we add the edges $xa_x,\ xb_x$ and $a_x b_x$.
        \item Finally, we add three more edges $t_1 t,\ t_2 t$ and $t_1 t_2$.
    \end{enumerate}

    \noindent
Note that $E(G) = \{sv, vt_1, vt_2 \mid v \in V(H)\} \cup \{xa_x, xb_x, a_xb_x \mid x \in X \cup \{s\}\} \cup \{t_1t, t_2t, t_1t_2\}$. 
Observe that the set $Y = \{ s,t_1,t_2\} \cup X \cup \{a_x \mid x \in X \cup \{s\} \}$ is a minimum vertex cover of $G$ of size $2|X|+4$.
It remains to construct the conflict graph $\widehat{G}$ such that $V(\widehat{G}) = E(G)$.
    We define the edge set $E(\widehat{G})$ as follows:
\begin{enumerate}[(i)]
        \item For every edge $uv \in E(H)$, we add an edge between $su$ and $sv$ in $\widehat{G}$.
        \item For every $v \in V(H)$, we add an edge between $sv$ and $vt_1$ in $\widehat{G}$.
        \item For every $x \in X \cup \{s\}$ the edges $xa_x,\ xb_x$ and $a_x b_x$ are pairwise adjacent in $\widehat{G}$.
        \item Finally, the edges $t_1 t,\ t_2 t$ and $t_1 t_2$ are pairwise adjacent in $\widehat{G}$.
    \end{enumerate}

    \noindent
    We set the target parameters to $k' = 2|X| + 4 \leq 4{\sf dtc} + 4$ and $\ell' = 2n - \ell$, where $n = |V(H)|$.
    
    \noindent
    It remains to prove that $H$ has an independent set of size of at least $\ell$ if and only if there exists a conflict-free cut in $(G,\widehat{G})$ of size at most $\ell'$. 
    First, we prove the following two structural claims about every conflict-free cut of $(G, \widehat{G})$.

    \begin{claim}
    \label{claim:intra-partite}
    The sets $\{t_1, t_2, t\}$ and for every $x \in X \cup \{s\}$, $\{x, a_x, b_x\}$  are intra-partite. 
    \end{claim}

    \begin{claimproof}
    By construction, the vertices $\{t_1, t_2, t\}$ induces a triangle in $G$ and the corresponding edges forms a clique in $\widehat{G}$.
    Similarly, for every $x \in X \cup \{s\}$, the vertices $\{x, a_x, b_x\}$ are pairwise adjacent in $\widehat{G}$.
    Due to \Cref{lemma:triangle-double-edge-monochromatic}, $\{t_1, t_2, t\}$ and $\{x, a_x, b_x\}$ are intra-partite.
    \end{claimproof}

    \begin{claim}
    \label{claim:inter-partite}
    The two sets $\{s\}$ and $\{t_1, t_2, t\}$ are inter-partite.
    \end{claim}

    \begin{claimproof}
    Let $(A,B)$ be any non-trivial vertex partition of $V(G)$ corresponding to a valid conflict-free edge cut.
    Suppose, for the sake of contradiction, that $\{s\}$ and $\{t_1, t_2, t\}$ lie in the same partition, say $A$.
    Since $B$ is nonempty, as a consequence of \Cref{claim:intra-partite} there exists $v \in V(H)$ such that $v \in B$.
    Therefore, the edges $sv$ and $vt_1$ belong to the edge cut, contradicting the fact that the edge cut is conflict free as $\{sv, vt_1\} \in E(\widehat{G})$.
    Hence, the claim follows.
    \end{claimproof}

	Now, we prove that the two instances $(H, \ell)$ and $(G, \widehat{G}, \ell')$.
	
    \begin{claim}
    \label{claim:independent-CFC-eqv}
    $H$ has an independent set of size at least $\ell$ if and only if $G$ contains a conflict-free cut of size at most $\ell'$.
    \end{claim}

    \begin{claimproof}
    First, we give the backward direction $(\Leftarrow)$ of the proof.
    Suppose that $(G, \widehat{G})$ contains a conflict-free edge cut $F$ of size at most $\ell'$.
    Let $(A, B)$ be the corresponding vertex partition. 
    By \Cref{claim:inter-partite}, we can assume without loss of generality that $s \in A$ and $t_1, t_2, t \in B$. Now \Cref{claim:intra-partite} ensures that no edge of any induced triangle in $G$ can be part of the cut $F$.
    In particular, for every $x \in X \cup \{s\}$, the edges $xa_x,\ xb_x$ and $a_x b_x$ are not in $F$. 
    Also, the edges $t_1 t,\ t_2 t$ and $t_1t_2$ are not in $F$. 
    Therefore, the only edges of $G$ that can be part of the cut $F$ are $\{sv, vt_1, vt_2 \mid v \in V(H)\}$. 
    
    Now, every vertex $v \in V(H)$ must be placed in either $A$ or $B$. If $v \in B$, the edge $sv \in F$. If $v \in A$, the edges $vt_1, vt_2 \in F$. Let $I = \{v \in V(H) \mid v \in B\}$.
    Observe that $I \neq \emptyset$, otherwise $E_G(A,B) = \{vt_1,\ vt_2 \mid v \in V(H)\}$ which is conflict free, but its size is $2n$ which contradicts the assumption that $|F| \leq \ell'$.
    The total size of $F$ is exactly $|F| = |I| + 2(n - |I|) = 2n - |I|$.
    Given that $|F| \leq \ell' = 2n - \ell$, it algebraically follows that $|I| \geq \ell$.
    Because $F$ is a conflict-free cut, no two edges from $\{sv \mid v \in I\}$ can be adjacent in $\widehat{G}$.
    By the construction of $\widehat{G}$, this strictly implies no two vertices in $I$ share an edge in $H$.
    Therefore, $I$ is a valid independent set of size at least $\ell$ in $H$.

    For the forward direction $(\Rightarrow)$, suppose that $H$ has an independent set $I$ of size at least $\ell$.
    We construct a non-trivial partition $(A, B)$ of $V(G)$ by defining 
    $B = I \cup \{t_1, t_2, t\} \cup \{a_x, b_x \mid x \in X \cap I\}$ and setting $A = V(G) \setminus B$.
    Therefore, the edge cut corresponding to the partition $(A, B)$ is $F = E_G(A,B) = \{sv \mid v \in I\} \cup \{vt_1, vt_2 \mid v \in V(H) \setminus I\}$.
    The total size of the cut-set is $|F| = |I| + 2(n - |I|) = 2n - |I|$. 
    Since $|I| \geq \ell$, it follows that $|F| \leq 2n - \ell = \ell'$. 
    
    It remains to verify that $F$ is an independent set in the conflict graph $\widehat{G}$. Observe that, by construction, $\{vt_1, vt_2 \mid v \in V(H) \setminus I\}$ is independent in $\widehat{G}$.
    The only edges in $F$ incident to $s$ are $\{sv \mid v \in I\}$. Since $I$ is an independent set in $H$, there is no edge $uv \in E(H)$ such that both $u, v \in I$.
    Consequently, no two edges in $\{sv \mid v \in I\}$ trigger a conflict. 
    Furthermore, $\widehat{G}$ strictly conflicts $sv$ with $vt_1$ for the same vertex $v$.
    By our partition, for any $v \in V(H)$, $F$ contains either $sv$ (when $v \in I$) or $vt_1$ (when $v \notin I$), but never both simultaneously.
    Therefore, $F$ is an independent set in $\widehat{G}$, and thus a valid conflict-free cut of size at most $\ell'$.
    \end{claimproof}
    
    As ${\sf vc}(G) \leq 4{\sf dtc} + 4$.
    This completes the proof of lemma.
\end{proof}

Now, we use this lemma to prove our final result that we restate below.

{\MinCutVC*}

\begin{proof}
Due to \Cref{lemma:min-cf-cut-vc}, given an instance $(H, \ell)$ of {\sc Independent Set}, there exists a polynomial-time algorithm that constructs an instance $(G, \widehat{G}, 2n - \ell)$ where $|V(H)| = n$.
In addition, it also satisfies that the ${\sf vc}(G) \leq 4{\sf dtc} + 4$.
Hence, this satisfies the requirement of a polynomial parameter transformation.
Since {\sc Independent Set} parameterized by ${\sf dtc}$ admits no polynomial kernel (unless {\nka}), and \Cref{lemma:min-cf-cut-vc} provides a polynomial parameter transformation, hence by \Cref{prop:no-poly-compression}, {\mincfcut} does not admit any polynomial kernel unless {\nka}.
\end{proof}


\section{Conclusions and Future Research}
\label{sec:conclusions}

Our paper makes significant progress towards providing a comprehensive picture on the fixed-parameterized tractability and kernelization complexity of {\CFCut}. 
A similar variant can also be defined for the vertex variant {\sc Conflict-Free Vertex Cut} in which the objective is to find a vertex cut that is independent in the conflict graph.
We believe that some of our (tractability/hardness) results can be extended for the ``vertex variant'' (in which the objective is to find a vertex cut that is independent in the conflict graph) of this problem as well.
It would be interesting to identify a situation in which the parameterized complexity status of {\CFCut} behaves differently from {\sc Conflict-Free Vertex Cut}.
This itself provides a direction of future research.
While our result proves that unless {\nka}, {\mincut} admits no polynomial kernel when the vertex cover number is the parameter, but that does not automatically translate to {\cfcut}.
Hence, the polynomial kernelization status of {\cfcut} when parameterized by the vertex cover number (${\sf vc}$) remains an open problem.
Additionally, {\sc Max-Cut} is NP-Complete problem and it would be interesting to investigate the conflict-free variant of {\sc Max Cut} from the perspective of parameterized complexity as well.

\subparagraph*{AI Disclosure:} The ideas of the algorithms, hardness reductions, and correctness proof arguments are all contributions of the authors themselves. 
The proofs ideas of \Cref{thm:diss-hardness} and \Cref{lemma:min-cf-cut-vc} were scaffolded to Gemini-Pro tool to provide the LaTeX code and it was subsequently carefully polished by the authors.
For the rest of the results, the proof ideas and subsequently, writing and polishing of the proofs are contributions of the authors themselves.
Apart from the above, Claude Pro was used for polishing texts in various parts of the paper.


\bibliography{refs.bib}

@inproceedings{golovach2026identification,
  author    = {Golovach, Petr A. and Morelle, Laure and Paulusma, Dani{\"e}l},
  title     = {Identification to Subclasses of Chordal Graphs},
  booktitle = {21st International Symposium on Parameterized and Exact Computation (IPEC 2026)},
  year      = {2026},
  note      = {To appear}
}

@inproceedings{BodlaenderHvL2025,
  author    = {Bodlaender, Hans L. and Hermelin, Danny and van Leeuwen, Erik Jan},
  title     = {Concurrency Constrained Scheduling with Tree-Like Constraints},
  booktitle = {Graph-Theoretic Concepts in Computer Science -- 51st International Workshop (WG 2025)},
  series    = {Lecture Notes in Computer Science},
  volume    = {16124},
  pages     = {105--120},
  year      = {2025},
  publisher = {Springer}
}

@article{TellacheB2018,
  author  = {Tellache, Nour El Houda and Boudhar, Mourad},
  title   = {Flow Shop Scheduling Problem with Conflict Graphs},
  journal = {Annals of Operations Research},
  volume  = {261},
  number  = {1},
  pages   = {339--363},
  year    = {2018},
  publisher = {Springer},
  doi     = {10.1007/s10479-017-2560-x}
}

@article{MallekBoudhar2024,
  author  = {Mallek, Amin and Boudhar, Mourad},
  title   = {Scheduling on Uniform Machines with a Conflict Graph: Complexity and Resolution},
  journal = {International Transactions in Operational Research},
  volume  = {31},
  number  = {2},
  pages   = {863--888},
  year    = {2024}
}

@article{ZhangZCCW2023,
  author  = {Zhang, An and Zhang, Liang and Chen, Yong and Chen, Guangting and Wang, Xing},
  title   = {Complexity and Approximation Algorithms for Two Parallel Dedicated Machine Scheduling with Conflict Constraints},
  journal = {Theoretical Computer Science},
  volume  = {941},
  pages   = {167--179},
  year    = {2023},
  publisher = {Elsevier},
  doi     = {10.1016/j.tcs.2022.11.012}
}

@article{HongBertrand2018,
  author  = {Hong, Huai-Che and Lin, Bertrand M. T.},
  title   = {Parallel Dedicated Machine Scheduling with Conflict Graphs},
  journal = {Computers \& Industrial Engineering},
  volume  = {124},
  pages   = {316--321},
  year    = {2018}
}

@article{BodlaenderJW94,
  author  = {Bodlaender, Hans L. and Jansen, Klaus and Woeginger, Gerhard J.},
  title   = {Scheduling with Incompatible Jobs},
  journal = {Discrete Applied Mathematics},
  volume  = {55},
  number  = {3},
  pages   = {219--232},
  year    = {1994}
}

@article{ChitnisCHPP16,
  author  = {Chitnis, Rajesh and Cygan, Marek and Hajiaghayi, MohammadTaghi and Pilipczuk, Marcin and Pilipczuk, Micha{\l}},
  title   = {Designing FPT Algorithms for Cut Problems Using Randomized Contractions},
  journal = {SIAM Journal on Computing},
  volume  = {45},
  number  = {4},
  pages   = {1171--1229},
  year    = {2016}
}

@article{CyganPW2013,
  author  = {Cygan, Marek and Pilipczuk, Marcin and Pilipczuk, Micha{\l} and Wojtaszczyk, Jakub Onufry},
  title   = {On Multiway Cut Parameterized above Lower Bounds},
  journal = {ACM Transactions on Computation Theory},
  volume  = {5},
  number  = {1},
  pages   = {3:1--3:11},
  year    = {2013}
}

@article{MarxRazgon2014,
  author  = {Marx, D{\'a}niel and Razgon, Igor},
  title   = {Fixed-Parameter Tractability of Multicut Parameterized by the Size of the Cutset},
  journal = {SIAM Journal on Computing},
  volume  = {43},
  number  = {2},
  pages   = {355--388},
  year    = {2014}
}

@article{ChenLL2009,
  author  = {Chen, Jianer and Liu, Yang and Lu, Songjian},
  title   = {An Improved Parameterized Algorithm for the Minimum Node Multiway Cut Problem},
  journal = {Algorithmica},
  volume  = {55},
  number  = {1},
  pages   = {1--13},
  year    = {2009}
}

@article{FariaKSSS20,
  author       = {Lu{\'{e}}rbio Faria and
                  Sulamita Klein and
                  Ignasi Sau and
                  U{\'{e}}verton S. Souza and
                  Rubens Sucupira},
  title        = {Maximum cuts in edge-colored graphs},
  journal      = {Discret. Appl. Math.},
  volume       = {281},
  pages        = {229--234},
  year         = {2020}
}

@article{GomesS21,
  author       = {Guilherme de C. M. Gomes and
                  Ignasi Sau},
  title        = {Finding Cuts of Bounded Degree: Complexity, {FPT} and Exact Algorithms,
                  and Kernelization},
  journal      = {Algorithmica},
  volume       = {83},
  number       = {6},
  pages        = {1677--1706},
  year         = {2021}
}

@article{ZhangKP2011,
  author  = {Zhang, Ruonan and Kabadi, Santosh N. and Punnen, Abraham P.},
  title   = {The Minimum Spanning Tree Problem with Conflict Constraints and Its Variations},
  journal = {Discrete Optimization},
  volume  = {8},
  number  = {2},
  pages   = {191--205},
  year    = {2011}
}

@article{GendreauLS2004,
  author  = {Gendreau, Michel and Laporte, Gilbert and Semet, Fr{\'e}d{\'e}ric},
  title   = {Heuristics and Lower Bounds for the Bin Packing Problem with Conflicts},
  journal = {Computers \& Operations Research},
  volume  = {31},
  number  = {3},
  pages   = {347--358},
  year    = {2004},
  publisher = {Elsevier},
  doi     = {10.1016/S0305-0548(02)00195-8}
}

@article{CapuaFLV18,
  author  = {Capua, Renatha and Frota, Yuri and Ochi, Luiz Satoru and Vidal, Thibaut},
  title   = {A Study on Exponential-Size Neighborhoods for the Bin Packing Problem with Conflicts},
  journal = {Journal of Heuristics},
  volume  = {24},
  number  = {4},
  pages   = {667--695},
  year    = {2018},
  publisher = {Springer},
  doi     = {10.1007/s10732-018-9372-2}
}

@article{BettinelliCM17,
  author  = {Bettinelli, Andrea and Cacchiani, Valentina and Malaguti, Enrico},
  title   = {A Branch-and-Bound Algorithm for the Knapsack Problem with Conflict Graph},
  journal = {INFORMS Journal on Computing},
  volume  = {29},
  number  = {3},
  pages   = {457--473},
  year    = {2017}
}

@article{DarmannPSW11,
  author       = {Andreas Darmann and
                  Ulrich Pferschy and
                  Joachim Schauer and
                  Gerhard J. Woeginger},
  title        = {Paths, trees and matchings under disjunctive constraints},
  journal      = {Discret. Appl. Math.},
  volume       = {159},
  number       = {16},
  pages        = {1726--1735},
  year         = {2011}
}

@article{gupta2026computational,
  author  = {Gupta, Shiwali and Mathew, Rogers},
  title   = {Computational and Combinatorial Results on Conflict-free Choosability},
  journal = {arXiv preprint arXiv:2605.10776},
  year    = {2026}
}

@article{BodlaenderKP2021,
  author  = {Bodlaender, Hans L. and Kolay, Sudeshna and Pieterse, Astrid},
  title   = {Parameterized Complexity of Conflict-Free Graph Coloring},
  journal = {SIAM Journal on Discrete Mathematics},
  volume  = {35},
  number  = {3},
  pages   = {2003--2038},
  year    = {2021}
}

@article{BhyravarapuHHKR24,
  author  = {Bhyravarapu, Sriram and Hartmann, Tim A. and Hoang, Hung P. and Kalyanasundaram, Subrahmanyam and Reddy, I. Vinod},
  title   = {Conflict-Free Coloring: Graphs of Bounded Clique Width and Intersection Graphs},
  journal = {Algorithmica},
  volume  = {86},
  number  = {7},
  pages   = {2250--2274},
  year    = {2024}
}

@inproceedings{AgrawalJKL018,
  author    = {Akanksha Agrawal and
               Pallavi Jain and
               Lawqueen Kanesh and
               Daniel Lokshtanov and
               Saket Saurabh},
  editor    = {Igor Potapov and
               Paul G. Spirakis and
               James Worrell},
  title     = {Conflict Free Feedback Vertex Set: {A} Parameterized Dichotomy},
  booktitle = {43rd International Symposium on Mathematical Foundations of Computer
               Science, {MFCS} 2018, Liverpool, UK, August 27-31, 2018},
  series    = {LIPIcs},
  volume    = {117},
  pages     = {53:1--53:15},
  publisher = {Schloss Dagstuhl - Leibniz-Zentrum f{\"{u}}r Informatik},
  year      = {2018},
  url       = {https://doi.org/10.4230/LIPIcs.MFCS.2018.53},
  doi       = {10.4230/LIPICS.MFCS.2018.53},
  bibsource = {dblp computer science bibliography, https://dblp.org}
}

@article{AgrawalJKS20,
  author    = {Akanksha Agrawal and
               Pallavi Jain and
               Lawqueen Kanesh and
               Saket Saurabh},
  title     = {Parameterized Complexity of Conflict-Free Matchings and Paths},
  journal   = {Algorithmica},
  volume    = {82},
  number    = {7},
  pages     = {1939--1965},
  year      = {2020},
  url       = {https://doi.org/10.1007/s00453-020-00681-y},
  doi       = {10.1007/S00453-020-00681-Y},
  bibsource = {dblp computer science bibliography, https://dblp.org}
}

@inproceedings{AravindS21,
  author    = {N. R. Aravind and
               Roopam Saxena},
  editor    = {Paola Flocchini and
               Lucia Moura},
  title     = {An {FPT} Algorithm for Matching Cut and d-Cut},
  booktitle = {Combinatorial Algorithms - 32nd International Workshop, {IWOCA} 2021,
               Ottawa, ON, Canada, July 5-7, 2021, Proceedings},
  series    = {Lecture Notes in Computer Science},
  volume    = {12757},
  pages     = {531--543},
  publisher = {Springer},
  year      = {2021},
  url       = {https://doi.org/10.1007/978-3-030-79987-8\_37},
  doi       = {10.1007/978-3-030-79987-8\_37},
  bibsource = {dblp computer science bibliography, https://dblp.org}
}

@inproceedings{BanikFG0J024,
  author    = {Aritra Banik and
               Fedor V. Fomin and
               Petr A. Golovach and
               Tanmay Inamdar and
               Satyabrata Jana and
               Saket Saurabh},
  editor    = {Timothy M. Chan and
               Johannes Fischer and
               John Iacono and
               Grzegorz Herman},
  title     = {Cuts in Graphs with Matroid Constraints},
  booktitle = {32nd Annual European Symposium on Algorithms, {ESA} 2024, Royal Holloway,
               London, United Kingdom, September 2-4, 2024},
  series    = {LIPIcs},
  volume    = {308},
  pages     = {17:1--17:15},
  publisher = {Schloss Dagstuhl - Leibniz-Zentrum f{\"{u}}r Informatik},
  year      = {2024},
  url       = {https://doi.org/10.4230/LIPIcs.ESA.2024.17},
  doi       = {10.4230/LIPICS.ESA.2024.17},
  bibsource = {dblp computer science bibliography, https://dblp.org}
}

@inproceedings{bodlaender_cross-composition_2010,
  author    = {Hans L. Bodlaender and
               Bart M. P. Jansen and
               Stefan Kratsch},
  editor    = {Thomas Schwentick and
               Christoph D{\"{u}}rr},
  title     = {Cross-Composition: {A} New Technique for Kernelization Lower Bounds},
  booktitle = {28th International Symposium on Theoretical Aspects of Computer Science,
               {STACS} 2011, Dortmund, Germany, March 10-12, 2011},
  series    = {LIPIcs},
  volume    = {9},
  pages     = {165--176},
  publisher = {Schloss Dagstuhl - Leibniz-Zentrum f{\"{u}}r Informatik},
  year      = {2011},
  url       = {https://doi.org/10.4230/LIPIcs.STACS.2011.165},
  doi       = {10.4230/LIPICS.STACS.2011.165},
  bibsource = {dblp computer science bibliography, https://dblp.org}
}

@article{comas_vertex_2009,
  title   = {Vertex fusion under distance constraints},
  volume  = {30},
  issn    = {0195-6698},
  url     = {https://www.sciencedirect.com/science/article/pii/S019566980900064X},
  doi     = {https://doi.org/10.1016/j.ejc.2009.03.014},
  number  = {7},
  journal = {European Journal of Combinatorics},
  author  = {Comas, Marc and Serna, Maria},
  year    = {2009},
  pages   = {1612--1623}
}

@book{cygan2015parameterized,
  title     = {Parameterized algorithms},
  author    = {Cygan, Marek and Fomin, Fedor V and Kowalik, {\L}ukasz and Lokshtanov, Daniel and Marx, D{\'a}niel and Pilipczuk, Marcin and Pilipczuk, Micha{\l} and Saurabh, Saket},
  volume    = {5},
  year      = {2015},
  publisher = {Springer}
}

@book{diestel2025graph,
  author    = {Reinhard Diestel},
  title     = {Graph Theory, 4th Edition},
  series    = {Graduate texts in mathematics},
  volume    = {173},
  publisher = {Springer},
  year      = {2012},
  isbn      = {978-3-642-14278-9},
  bibsource = {dblp computer science bibliography, https://dblp.org}
}

@book{downey2013fundamentals,
  title     = {Fundamentals of parameterized complexity},
  author    = {Downey, Rodney G and Fellows, Michael R and others},
  volume    = {4},
  year      = {2013},
  publisher = {Springer}
}

@article{FeghaliLPR25,
  author    = {Carl Feghali and
               Felicia Lucke and
               Dani{\"{e}}l Paulusma and
               Bernard Ries},
  title     = {Matching Cuts in Graphs of High Girth and H-Free Graphs},
  journal   = {Algorithmica},
  volume    = {87},
  number    = {8},
  pages     = {1199--1221},
  year      = {2025},
  url       = {https://doi.org/10.1007/s00453-025-01318-8},
  doi       = {10.1007/S00453-025-01318-8},
  bibsource = {dblp computer science bibliography, https://dblp.org}
}

@article{fomin2019subquadratic,
  title     = {Subquadratic kernels for implicit 3-hitting set and 3-set packing problems},
  author    = {Fomin, Fedor V and Le, Tien-Nam and Lokshtanov, Daniel and Saurabh, Saket and Thomass{\'e}, St{\'e}phan and Zehavi, Meirav},
  journal   = {ACM Transactions on Algorithms (TALG)},
  volume    = {15},
  number    = {1},
  pages     = {1--44},
  year      = {2019},
  publisher = {ACM New York, NY, USA}
}

@article{JacobMR21,
  author    = {Ashwin Jacob and
               Diptapriyo Majumdar and
               Venkatesh Raman},
  title     = {Parameterized Complexity of Conflict-Free Set Cover},
  journal   = {Theory Comput. Syst.},
  volume    = {65},
  number    = {3},
  pages     = {515--540},
  year      = {2021},
  url       = {https://doi.org/10.1007/s00224-020-10022-9},
  doi       = {10.1007/S00224-020-10022-9},
  bibsource = {dblp computer science bibliography, https://dblp.org}
}

@article{JainKM20,
  author    = {Pallavi Jain and
               Lawqueen Kanesh and
               Pranabendu Misra},
  title     = {Conflict Free Version of Covering Problems on Graphs: Classical and
               Parameterized},
  journal   = {Theory Comput. Syst.},
  volume    = {64},
  number    = {6},
  pages     = {1067--1093},
  year      = {2020},
  url       = {https://doi.org/10.1007/s00224-019-09964-6},
  doi       = {10.1007/S00224-019-09964-6},
}

@inproceedings{KalyanasundaramK26,
  author    = {Subrahmanyam Kalyanasundaram and
               Subodh Kumar},
  editor    = {Florent Foucaud and
               Aline Parreau},
  title     = {Conflict-Free Cuts in Planar and 3-Degenerate Graphs with 1-Regular
               Conflicts},
  booktitle = {Combinatorial Algorithms - 37th International Workshop, {IWOCA} 2026,
               Clermont-Ferrand, France, June 8-11, 2026, Proceedings},
  series    = {Lecture Notes in Computer Science},
  volume    = {16587},
  pages     = {385--399},
  publisher = {Springer},
  year      = {2026}
}

@article{KomusiewiczKL20,
  author    = {Christian Komusiewicz and
               Dieter Kratsch and
               Van Bang Le},
  title     = {Matching cut: Kernelization, single-exponential time FPT, and exact
               exponential algorithms},
  journal   = {Discret. Appl. Math.},
  volume    = {283},
  pages     = {44--58},
  year      = {2020},
  url       = {https://doi.org/10.1016/j.dam.2019.12.010},
  doi       = {10.1016/J.DAM.2019.12.010},
  bibsource = {dblp computer science bibliography, https://dblp.org}
}

@article{LokshtanovPSSZ20,
  author    = {Daniel Lokshtanov and
               Fahad Panolan and
               Saket Saurabh and
               Roohani Sharma and
               Meirav Zehavi},
  title     = {Covering Small Independent Sets and Separators with Applications to
               Parameterized Algorithms},
  journal   = {{ACM} Trans. Algorithms},
  volume    = {16},
  number    = {3},
  pages     = {32:1--32:31},
  year      = {2020},
  url       = {https://doi.org/10.1145/3379698},
  doi       = {10.1145/3379698},
  bibsource = {dblp computer science bibliography, https://dblp.org}
}

@article{LuckePR23,
  author  = {Felicia Lucke and
             Dani{\"{e}}l Paulusma and
             Bernard Ries},
  title   = {Finding Matching Cuts in H-Free Graphs},
  journal = {Algorithmica},
  volume  = {85},
  number  = {10},
  pages   = {3290--3322},
  year    = {2023}
}

@article{LuckePR24,
  author    = {Felicia Lucke and
               Dani{\"{e}}l Paulusma and
               Bernard Ries},
  title     = {Dichotomies for Maximum Matching Cut: H-freeness, bounded diameter,
               bounded radius},
  journal   = {Theor. Comput. Sci.},
  volume    = {1017},
  pages     = {114795},
  year      = {2024},
  url       = {https://doi.org/10.1016/j.tcs.2024.114795},
  doi       = {10.1016/J.TCS.2024.114795},
  bibsource = {dblp computer science bibliography, https://dblp.org}
}

@article{majumdar_approximate_2020,
  title   = {On the {Approximate} {Compressibility} of {Connected} {Vertex} {Cover}},
  volume  = {82},
  issn    = {1432-0541},
  url     = {https://doi.org/10.1007/s00453-020-00708-4},
  doi     = {10.1007/s00453-020-00708-4},
  number  = {10},
  journal = {Algorithmica},
  author  = {Majumdar, Diptapriyo and Ramanujan, M. S. and Saurabh, Saket},
  month   = oct,
  year    = {2020},
  pages   = {2902--2926}
}

@article{morelle_vertex_2026,
  title    = {Vertex identification to a forest},
  volume   = {349},
  issn     = {0012-365X},
  url      = {https://www.sciencedirect.com/science/article/pii/S0012365X25003073},
  doi      = {https://doi.org/10.1016/j.disc.2025.114699},
  number   = {1},
  journal  = {Discrete Mathematics},
  author   = {Morelle, Laure and Sau, Ignasi and Thilikos, Dimitrios M.},
  year     = {2026},
  pages    = {114699}
}

@article{RauchRS25,
  author    = {Johannes Rauch and
               Dieter Rautenbach and
               U{\'{e}}verton S. Souza},
  title     = {On conflict-free cuts: Algorithms and complexity},
  journal   = {Inf. Process. Lett.},
  volume    = {187},
  pages     = {106503},
  year      = {2025},
  url       = {https://doi.org/10.1016/j.ipl.2024.106503},
  doi       = {10.1016/J.IPL.2024.106503},
  bibsource = {dblp computer science bibliography, https://dblp.org}
}

@article{StoerWagner1997,
  author    = {Mechthild Stoer and Frank Wagner},
  title     = {A Simple Min-Cut Algorithm},
  journal   = {Journal of the ACM},
  volume    = {44},
  number    = {4},
  pages     = {585--591},
  year      = {1997},
  doi       = {10.1145/263867.263872}
}

\end{document}